\documentclass[12pt,reqno]{amsart}

\usepackage[top=1.2in, bottom=1.2in, left=1.1in, right=1.1in]{geometry}
\usepackage{amsfonts, amsthm,amssymb, amsmath, mathtools, dsfont, xcolor}
\usepackage[linktoc=page, colorlinks, linkcolor=blue, citecolor=blue]{hyperref}
\usepackage{enumerate, commath}
\usepackage{algorithm,algpseudocode}

\usepackage{lipsum}

\usepackage{graphicx}

\usepackage{subcaption}
\usepackage[font=small]{caption}
\numberwithin{equation}{section}

\DeclareMathOperator{\E}{\mathbb{E}}

\DeclareMathOperator{\dist}{dist}
\DeclareMathOperator{\Span}{span}

\DeclareMathOperator{\supp}{supp}

\DeclareMathOperator{\argmin}{argmin}

\renewcommand{\Pr}[2][]{\mathbb{P}_{#1} \left\{ #2 \rule{0mm}{3mm}\right\}}

\newcommand{\ip}[2]{\langle#1,#2\rangle}
\newcommand{\Bigip}[2]{\Big\langle#1,#2\Big\rangle}

\def \N {\mathbb{N}}

\def \R {\mathbb{R}}

\def \a {\alpha}
\def \b {\beta}
\def \g {\gamma}
\def \e {\varepsilon}
\def \d {\delta}
\def \l {\lambda}

\def \one {{\textbf 1}}

\def \Id {\mathrm{Id}}

\newtheorem{theorem}{Theorem}[section]
\newtheorem{proposition}[theorem]{Proposition}

\newtheorem{lemma}[theorem]{Lemma}

\newtheorem{definition}[theorem]{Definition}

\theoremstyle{remark}
\newtheorem{remark}[theorem]{Remark}

\begin{document}

\title[Private sampling and synthetic data]{Private sampling: a noiseless approach for generating differentially private synthetic data}

\author{March Boedihardjo}
\address{Department of Mathematics, University of California Irvine}
\email{marchb@uci.edu}
\author{Thomas Strohmer}
\address{Center of Data Science and Artificial Intelligence Research University of California, Davis \\ and Department of Mathematics, University of California Davis}
\email{strohmer@math.ucdavis.edu}
\author{Roman Vershynin}
\address{Department of Mathematics, University of California Irvine}
\email{rvershyn@uci.edu}

\maketitle

\begin{abstract}
In a world where artificial intelligence and data science become omnipresent, data sharing is increasingly locking horns with data-privacy concerns.
Differential privacy has emerged as a rigorous framework for protecting individual privacy in a statistical database, while releasing  useful  statistical  information  about  the  database. The standard way to implement differential privacy is to inject a sufficient amount of noise into the data. However, in addition to other limitations of differential privacy, this process of adding noise will affect data accuracy and utility.  Another approach to enable privacy in data sharing is based on the concept of synthetic data. The goal of synthetic data is  to create an as-realistic-as-possible dataset, one that not only maintains the nuances of the original data, but does so without risk of exposing sensitive information.
The combination of differential privacy with synthetic data has been suggested as a  best-of-both-worlds solutions. In this work, we propose the first noisefree method to construct differentially private synthetic data; we do this through a mechanism called ``private sampling''.  Using the Boolean cube as benchmark data model,  we derive explicit bounds on accuracy and privacy of the constructed synthetic data. The key mathematical tools are hypercontractivity, duality, and empirical processes. A core ingredient of our private sampling mechanism is a rigorous  ``marginal correction'' method, which has the remarkable property that importance reweighting can be utilized to exactly match the marginals of the sample to the marginals of the population.

\end{abstract}

\section{Introduction}\label{s:intro}

In a world where artificial intelligence and data science are penetrating more and more aspects of our life, data sharing is increasingly locking horns with data-privacy concerns. This conflict  is playing out around the globe, as private and public organizations are trying to find ways to share data without compromising sensitive personal information.

There exist various attempts to protect sensitive information in data.
Historically the way to share private information without betraying privacy was through {\em anonymization}~\cite{zhou2008brief},
i.e., by stripping away enough identifying information from a dataset, so that the so-modified data could be shared freely.  Anonymization, however, proved to be a fragile means to protect data privacy.
In actuality, identifying individuals using seemingly non-unique identifiers is far easier than proponents of data anonymization expected.
For instance, Netflix and AOL customers were all accurately identified from purportedly anonymized data.
De-identification requires precise definitions of ``unique identifiers''. Furthermore, de-identification suffers from an aging problem: it is already quite  difficult enough to determine exactly what data identifies information that needs to be protected (say, the identity of individuals), but it is even more difficult to accurately predict what potential auxiliary information could be available in the
future.  This leads to an arms race between  de-identification and re-identification.

The well-documented failures of anonymization have prompted aggressive research on data sanitization, ranging from $k$-anonymity~\cite{sweeney2002k,bayardo2005data} to today’s highly acclaimed differential privacy~\cite{dwork2014algorithmic}. The concept of k-anonymity  was introduced to address the risk of re-identification of anonymized data through linkage to other datasets. The idea behind {\em $k$-anonymity} is to maintain privacy by guaranteeing that for every record in a database there  are $k$ of indistinguishable copies.

{\em Differential privacy} is a framework to quantify the extent to which individual privacy in a statistical database is preserved while releasing  useful  statistical  information  about  the  database~\cite{dwork2014algorithmic}.  Differential privacy is a popular and robust method that  comes with a rigorous mathematical framework and provable guarantees.
Differential privacy can protect aggregate information, but not sensitive information in general. Also, if enough identical queries are asked, the protection provided by differential privacy is diluted. Additionally, if the query being asked requires high specificity, then it is more difficult to uphold differential privacy. In any case,  in all the aforementioned methods the basic tradeoff between utility and privacy represents a serious limitation.

{\em Synthetic data}  provide a promising concept  to solve this conundrum~\cite{bellovin2019privacy}. The goal of synthetic data
is  to create an as-realistic-as-possible dataset, one that not only maintains the nuances of the original data, but does so without risk of exposing sensitive information. Synthetic datasets are generated from existing datasets and maintain the statistical properties of the original dataset. Since (ideally) synthetic data contain no protected information, the datasets can be shared freely
among investigators in academia or industry, without security and privacy concerns.

It has been frequently recommended that synthetic data may be combined with differential privacy to achieve a best-of-both-worlds scenario~\cite{hardt2012simple,bellovin2019privacy,kearnsroth2020,liu2021leveraging,BSV2021a}. As observed in~\cite{bellovin2019privacy}, ``The most ideal data to use in any analysis will always be original data. But when that option is not available, synthetic data plus differential privacy offers a great compromise.''  Synthetic data  are not only a succinct way of representing the
answers to large numbers of queries, but they also permit one to carry out other data analysis tasks, such as visualization or regression.

The standard way to achieve differential privacy is to add noise, either to the data queries, the data themselves, or in case of synthetic data during the data generation process, for a small sample of work see e.g.~\cite{dwork2014algorithmic,hardt2012simple,hardt2010multiplicative,aydore2021differentially,liu2021leveraging,de2021compressive}.
Unfortunately, noise will negatively affect utility and can inject systematic errors---hence bias---into the data~\cite{pujol2020fair,zhu2020bias,fioretto2021decision}. To illustrate these issues, assume the dataset under consideration consists of images, each depicting the face of a person.
We can attempt to generate a differentially private synthetic dataset by adding a sufficient amount of noise to each image (e.g., by adding random noise~\cite{newton2005preserving} or by distorting or blurring the images~\cite{ren2018learning,yang2021study}), such that the persons in the images can no longer be identified. Ignoring for the moment the possibility of re-identifying a person by applying denoising or deblurring techniques to the distorted images, it is clear that utility of this dataset can decrease significantly during this process of adding noise, perhaps to the point that many of the nuances one might be interested in are no longer present.

To illuminate the effect of introducing systematic error when adding noise to ensure differential privacy, we just need to look at the issues reported with  differentially private US Census 2020 demonstration data, which have resulted in diminished quality of statistics for small populations such as tribal nations~\cite{wezerek,pujol2020fair,fioretto2021decision}.

These considerations raise a fundamental question:

\vspace*{2mm}
\centerline{\em Can we  generate differentially private synthetic data without adding noise?}

\vspace*{2mm}
In this paper, we give a positive and constructive answer. Using the Boolean cube as our data model, we will develop a noiseless method to generate synthetic data, which approximately preserve low-dimensional marginals of the original dataset. Our method is based on a {\em private sampling} framework and comes with explicit bounds on privacy and accuracy.  The key mathematical tools are hypercontractivity, duality, and empirical processes.
 A core ingredient of our private sampling framework is a rigorous  ``marginal correction'' method, which has the remarkable property that importance reweighting can be utilized to {\em exactly} match the marginals of the sample to the marginals of the population.

There exist other methods to generate differentially private synthetic data without adding noise, such as those based on generative adversarial networks~\cite{lu2017poster,abay2018privacy,brock2018large,zhu2019electrocardiogram,delaney2019synthesis}. However, these methods are just empirical and do not come with any rigorous bounds regarding accuracy or privacy. Those deep learning based methods that do come with privacy guarantees---but still without any accuracy guarantees---require injecting noise into the synthetic data generation process~\cite{xie2018differentially,jordon2018pate,beaulieu2019privacy}.

\section{Synthetic data and differential privacy}\label{s:dp}

Differential  privacy  has  emerged as  the  de  facto  standard  for  guaranteeing  privacy  in data sharing. Recall the definition of differential privacy:

\begin{definition}[Differential Privacy~\cite{dwork2014algorithmic}]  {\em A randomized mechanism ${\mathcal M}:  {\mathcal S}^N \to {\mathcal R}$ satisfies $\e$-differential privacy
if for any two adjacent datasets $X_1, X_2 \in {\mathcal S}^N$ differing by one element,  and any output subset  ${\mathcal O} \in {\mathcal R}$ it holds that}
$${\mathbb P} [{\mathcal M}(X_1) \in {\mathcal O}] \le e^\e \cdot {\mathbb P} [{\mathcal M}(X_2) \in {\mathcal O}].$$
\end{definition}

Numerous techniques have been proposed for generating privacy-preserving synthetic data (e.g.~\cite{abowd2001disclosure,burridge2003information,abay2018privacy,dahmen2019synsys,mendelevitch2021fidelity}), but without providing formal privacy guarantees.
Almost all existing mechanisms to implement differential privacy inject some sort of noise 
into the data or the data queries, see e.g.\ the Laplacian mechanism~\cite{dwork2006calibrating}. This is also the case for differentially private synthetic data, see for instance~\cite{li2014differentially,barak2007privacy}.

Obviously, we want our synthetic data to be similar to the original data. To that end we need some metrics to measure similarity. A common and natural choice is to try to (approximately) preserve low-dimensional marginals~\cite{barak2007privacy,thaler2012faster}. A marginal of the data $X$ is the fraction of the elements $x_i$ with specified values of specified parameters. On the one hand, marginals are important in their own right as a tool of statistical analysis.
On the other hand, if the synthetic data preserve e.g.\ two-dimensional marginals (i.e., covariance matrices) with
sufficient accuracy,  the synthetic dataset is expected to inherit other significant properties from the original dataset,
such as similar  behavior with respect to clustering, classification or regression\footnote{So far this expectation has only
been verified empirically in various papers, while a rigorous mathematical verification is an important open problem.}.

However, we are immediately met with a remarkable {\em no-go} theorem  due to Ullman and Vadhan~\cite{ullman2011pcps}. They proved the surprising result that (under standard cryptographic assumptions) there is no polynomial-time differentially private algorithm that takes a dataset $X \in (\{0,1\}^p)^n$ and outputs a synthetic dataset $Y \in (\{0,1\}^p)^k$ such that all two-dimensional marginals of $Y$  are approximately equal to those of $X$.

There is an extensive literature on privately releasing answers to linear queries, but without producing synthetic data, see e.g.~\cite{hardttalwar,nikolov2013geometry,dwork2015efficient} for a small sample. Another line of important work deals with
with privacy-preserving data analysis in a statistical framework~\cite{duchi2018minimax,cai2019cost}; but they also are not concerned with synthetic data. The papers~\cite{barak2007privacy,hardt2010multiplicative,hardt2012simple,blum2013learning} propose a range of interesting methods for producing approximately accurate private synthetic data. However, the associated algorithms have running time
that is at least exponential in $p$.

Luckily, already a slightly relaxed formulation of the worst-case no-go result in~\cite{ullman2011pcps} already leads to positive results. For example, if we relax ``all marginals'' to ``most marginals'',  it is shown in~\cite{BSV2021a}  that there exists a polynomial-time differentially private algorithm generating synthetic data $Y \in (\{0,1\}^p)^k$ such that the error between the marginals of $Y$ and $X$ is small. Remarkably, the result does not only hold for two-dimensional marginals, but for marginals of {\em all dimensions}. If we relax ``worst data'' to ``typical data'', generating accurate differentially private  synthetic Boolean (or other domain constrained) data becomes tractable~\cite{BSV2021b}.

Yet,  in all the aforementioned papers differential privacy is achieved by adding noise during the data generation process. In this paper we propose an alternative, {\em noise-free}, mechanism called {\em private sampling}.

\section{Main result}\label{s:main}

We model the true data $X=(x_1,\ldots,x_n)$ as a sequence of $n$ points from the Boolean cube $\{0,1\}^p$,  which is a standard benchmark data model~\cite{barak2007privacy,ullman2011pcps,hardt2012simple,ping2017datasynthesizer,liu2021leveraging,benaim2020analyzing}. For example, $X$
might represent the health records of $n$ patients, where each health record consists of $p$ parameters. These parameters are $0/1$ numbers that represent the answers to the standard health history questionnaire, such as ``does the patient smoke?'', ``does the patient have diabetes?". We can also represent categorical data
(gender, occupation, etc.) or numerical data (by splitting them into intervals)  on the Boolean cube via binary or one-hot encoding.

We would like to manufacture a synthetic dataset $Y=(y_1,\ldots,y_k)$, another sequence of $k$ elements of the cube.
Our two desiderata are {\em privacy} and {\em accuracy}.
Specifically, we would like the synthetic data to be differentially private,
and all low-dimensional marginals of $Y$ to exactly or approximately match those of  $X$.

We recall that on the Boolean cube, a marginal of a function $f:\{-1,1\}^p \to \R$ is defined as a sum of values of $f$ on the points of the cube that have specified values of specified parameters. For example, a two-dimensional marginal of $f$ is $\sum_{x \in \{-1,1\}^p} f(x) \one_{\{x(1)=x(2)=1\}}(x)$. If $f$ is a density, a marginal can be interpreted as the probability that a random point $Z$ drawn from the cube according to $f$ has specified values of specified parameters; in the example below it is $\Pr{Z(1)=Z(2)=1}$. Marginals of the data $X=(x_1,\ldots,x_n)$ can be interpreted as marginals of the uniform density $f_n=\frac{1}{n} \sum_{i=1}^n \one_{x_i}$ on $X$.
An example of a two-dimensional marginal is the fraction of elements $x_i$ whose first and second parameters equal $1$, i.e.
$\frac{1}{n} \sum_{i=1}^n \one_{\{x_i(1)=x_i(2)=1\}}$. This could represent for example the number of  patients who smoke and have diabetes.

\medskip
Here we explore a new {\em noiseless} approach: take a new sample $S = (s_1,\ldots,s_m)$ uniformly from the cube, reweight $S$ to make the marginals match those of the true data $X$, and resample from the weighted sample $S$.

But is this even possible? Let us assume the dataset $X = (x_1,\ldots,x_n)$ is drawn from the cube independently and according to some unknown density. Draw a new sample $S = (s_1,\ldots,s_m)$ according to some known density, for example uniformly from the cube\footnote{Since the cardinality of $S$ will be chosen to be smaller than that of the dataset $X$, we call $S$ also the reduced space. 
}. Can we reweight $S$ so that the reweighted sample has approximately the same marginals as $X$? Note that there are precisely $\binom{p}{\le d}$ marginals of degree at most $d$, where $\binom{p}{\le d} := \binom{p}{0}+\binom{p}{1}+\cdots + \binom{p}{d}$.
Surprisingly, we can even match all marginals {\em exactly}.

Let us state it this result informally; a rigorous, non-asymptotic and more general statement is given in Theorem~\ref{thm: master}.

\begin{theorem}[Matching marginals]		\label{thm: two measures}
  Consider two regularly varying densities\footnote{A density $f$ is {\em regularly varying} if
	$\sup f(x)/f(y) = O(1)$
	where the supremum is over all points $x$ and $y$ in the cube.
	Our results are more general; as we will see shortly, the regularity assumption can be relaxed.}
  on the cube $\{0,1\}^p$,
  and draw two independent samples $X$ and $S$
  from the cube according to these two distributions.
  If $\min(\abs{X},\abs{S}) \gg e^{2d} \binom{p}{\le d}$, then with probability $1-o(1)$
  there exists a density on $S$ that has exactly the same
  marginals up to dimension $d$ as the uniform distribution on $X$.
\end{theorem}

\begin{remark}
  To match all $\binom{p}{\le d}$ marginals of dimension at most $d$, it makes sense to have at least as many
  data points. This explains the requirement on $n$ in the theorem heuristically
  (but not rigorously).
  The prefactor $e^{2d}$ is negligible compared to $\binom{p}{\le d}$ if $d \ll p$.
\end{remark}

As a ``non-example'' for Theorem~\ref{thm: two measures}, consider a probability measure supported
on the set of patients whose first parameter equals $0$, and a different probability measure supported on the set of patients whose first parameter equals $1$.
Then even a one-dimensional marginal -- the distribution of the first parameter -- will be different for $X$ and $Y$, no matter how $Y$ is reweighted.
This example shows that some form of regularity assumption will be required in the theorem.

The density $h^*$ on $S$ that is guaranteed by Theorem~\ref{thm: two measures} can be {\em computed efficiently}. Indeed, this task can be set up as a linear program with $\abs{S}$ variables (the values of the density on $S$), $\binom{p}{\le d}$ linear equations (to match the marginals to those of $X$), and $\abs{S}$ linear inequalities (to ensure the density is nonnegative on $S$).

Once this density $h^*$ is computed, we can generate synthetic data $Y = (y_1,\ldots,y_k)$ by drawing independent points from $S$ according to the density $h^*$.

\subsection{Private sampling}

Is such synthetic data $Y$ private? Here is a general tool that basically says: yes, $Y$ is private as long as the density $h^*$ has bounded sensitivity.

\begin{lemma}[Private sampling]		\label{lem: private sampling}
  Let $\Omega$ be a finite set.
  Let $f$ be a mapping that takes a dataset $X$ as input
  and returns a probability mass function $f(X)$ on $\Omega$.
  Suppose $\e>0$ and $k \in \N$ are chosen so that
  $$
  \norm{f(X_1)/f(X_2)}_\infty \le \exp(\e/k)
  $$
  for all datasets $X_1$ and $X_2$ that differ on a single element.
  Then the algorithm that takes $X$ as input and returns
  a sample of $k$ points drawn from $\Omega$ independently
  and according to the distribution $f(X)$ is $\e$-differentially private.
\end{lemma}

\begin{proof}
The probability that a given $k$-tuple of points $\omega_1,\ldots,\omega_k \in \Omega$ is drawn when sampled from distribution $f(X_1)$
equals $\prod_{i=1}^k f(X_1)(\omega_i)$.
Similarly, the probability that this same tuple is drawn when sampled
from distribution $f(X_2)$ equals $\prod_{i=1}^k f(X_2)(\omega_i)$.
If the databases $X_1$ and $X_2$ differ on a single element, the assumption implies
that the ratio of these probabilities is bounded by
$\prod_{i=1}^k \exp(\e/k) = \exp(\e)$. This means that the sampling mechanism
is $\e$-differentially private.
\end{proof}

\subsection{Difficulties and their resolution}

Unfortunately, the density $h^*$ guaranteed by Theorem~\ref{thm: two measures}
is too sensitive. Indeed, the sensitivity bound in Lemma~\ref{lem: private sampling} needs to be proved for {\em arbitrary} input data, while Theorem~\ref{thm: two measures} only works with high probability. For some input data $X$, a suitable density exists, and for another input data $Z$, no suitable density exists. Moving from $X$ toward $Z$ by changing one data point at a time, we can find a pair of datasets $X_1$ and $X_2$ that differ in a single data point so that the algorithm succeeds to find a density for $X_1$ and fails for $X_2$. This means that the algorithm is non-private.

The other issue is that there can be (and usually are) many suitable densities $h^*$. Which one to chose? How to devise a selection rule that upholds privacy?

In other words, we need to work around the possible non-existence and non-uniqueness of the solution. We resolve both issues here. To ensure existence, we employ {\em shrinking}: we move the solution space (the set of all functions on $S$, possibly negative-valued, that have the same marginals as $X$) toward the uniform density on $S$ until the resulting set contains a nonnegative function (thus a density).
For the selection rule, we choose the closest solution to the uniform density on $S$ in the $L^2$ metric.

Furthermore, while $S$ is chosen randomly, we do need $S$ to be  {\em well-conditioned}
in a sense that will be discussed in detail in Section~\ref{s:solutionspace}. At this point suffice it to say that (i)~the well-conditionedness of $S$ can be expressed in terms of a  bound on the smallest singular value $\sigma_{\min}(M)$ of the $m \times \binom{p}{\le d}$ matrix $M$ with entries $w(s)$, where $s \in S$ and $w$ is a Walsh function\footnote{See Section~\ref{s:fourier} for basic definitions related to Fourier analysis of the Boolean cube.} of degree at most $d$;
(ii)~the well-conditionedness of $M$ can be easily achieved and easily verified.

This leads us to the algorithm outlined in the next subsection.

\subsection{Algorithm}

We provide a high-level description of our proposed method in Algorithm~\ref{mainalgorithm1}.

\begin{algorithm}[h!]
\caption{Private sampling synthetic data algorithm}
\label{mainalgorithm1}
\begin{algorithmic}

\State {\bf Input:} a sequence $X$ of $n$ points in  $\{-1,1\}^p$ (true data); $m$:  cardinality of $S$;
$d$: the degree of the marginals to be matched; parameters $\delta, \Delta$ with $\Delta>  \delta > 0$.

\begin{enumerate}[1.]

\item Draw $m$ points from $\{-1,1\}^p$ independently and uniformly,
and call this set $S$ (reduced space).

\item Form the $m \times \binom{p}{\le d}$ matrix $M$ with entries $w(s)$, where $s \in S$ and $w$ is a Walsh function of degree at most $d$.
If the smallest  singular value of $M$ is bounded below by $\sqrt{m}/2e^d$, call $S$ well conditioned and proceed. Otherwise return ``Failure'' and stop.

\item Consider the affine space $H$ consisting of all densities on $S$ that have exactly the same marginals up to dimension $d$ as the true data $X$.

\item If necessary, shrink $H$ toward the uniform density on $S$ just so the resulting affine space $\tilde{H}$ contains a density that is lower bounded by $2\d/m$ and upper bounded by $(\Delta-\delta)/m$.

\item Among all densities in $\tilde{H}$ that are lower bounded by $\d/m$ and upper bounded by $\Delta/m$, pick one closest to the uniform density in the $L^2$ norm.

\end{enumerate}

\State {\bf Output:} a sequence $Y$ of $k$ points from $S$ according to this density.

\end{algorithmic}

\end{algorithm}

The well-conditionedness of $S$ in Algorithm~\ref{mainalgorithm1} defined via the condition $\sigma_{\min}(M) > \sqrt{m}/2e^d$ essentially says that the subsampled Walsh basis is almost orthogonal. The scaling $\sqrt{m}$ is natural: the entries of $M$ all have  absolute value $1$, hence the columns of $M$ have Euclidean norm $\sqrt{m}$. If we had $\sigma_{\min}(M) = \sqrt{m}$,  this would imply that  the columns of $M$  (the subsampled Walsh functions) are mutually orthogonal. We require a relaxed (by a factor $2e^d$) version of this orthogonality. 

What if $S$ fails the desired condition? We can simply resample $S$ until it is well conditioned. But this is only a useful strategy if the chances of success are sufficiently high. Under some mild conditions (see Section~\ref{s:solutionspace})  success happens with probability $>1/2$, hence the expected number or trials until success is $\le 2$. 
This way Algorithm~\ref{mainalgorithm1} succeeds deterministically, but its running time becomes random (albeit with the rather modest expected overhead time $\le 2$).

\bigskip

\begin{definition}\label{def:accuracy}
We say that the synthetic dataset $Y$ is {\em $\d$-accurate} if each of its marginals
up to degree (or dimension) $d$ is within $\d$ from the corresponding
marginal of the true dataset $X$.
\end{definition}

The following  theorem guarantees the accuracy and privacy of the algorithm. We state it informally here, and more accurately in Theorems~\ref{thm: privacy} and \ref{thm: accuracy}.

\begin{theorem}[Privacy and accuracy]		\label{thm: privacy-accuracy}
  Let the size of the reduced space $S$ satisfy $m \asymp e^{2d} \binom{p}{\le d}$.
  \begin{enumerate}[(a)]
    \item Algorithm~\ref{mainalgorithm1} succeeds (i.e. does not return ``Failure'') with high probability. 
    \item If the size of the synthetic data satisfies
       $k \ll \sqrt{n} /m$,
  	then Algorithm~\ref{mainalgorithm1} is $o(1)$-differentially private.
    \item Suppose $n \gg e^{2d} \binom{p}{\le d}$, $k \gg \log \binom{p}{\le d}$, and the true data points $X$ are sampled independently from some density that is upper bounded by $\Delta/2^p$. Then, with high probability, the synthetic data generated via Algorithm~\ref{mainalgorithm1} is $o(1)$-accurate up to dimension $d$.
  \end{enumerate}
\end{theorem}

For a more formal presentation of Algorithm~\ref{mainalgorithm1}, see Algorithm~\ref{mainalgorithm2} below. A formal version of part (a) of Theorem~\ref{thm: privacy-accuracy} is shown in Proposition~\ref{prop: success}; part (b) is shown in Theorem~\ref{thm: privacy} and Remark~\ref{rem: DP simplified}; part (c) is shown in Theorem~\ref{thm: accuracy}. The mathematical techniques to prove these results revolve around Fourier analysis of Boolean functions and empirical processes, see Sections~\ref{s:fourier}--\ref{s:lowdegree}.

In case the true data $X$ is sampled form a regular density, the algorithm will not apply any shrinkage, since in this case Theorem~\ref{thm: two measures} guarantees the existence of a solution. (We make this rigorous in Remark~\ref{rem: no shrinkage}.) In this case, the private synthetic data $Y$ will be sampled in an {\em unbiased way} from the density $h^*$ that has {\em exactly} the same marginals as the true data $X$.


\subsection{Further remarks}

There is a one-sample version of Theorem~\ref{thm: two measures}. Let us state it here informally; a more accurate statement is given in Theorem~\ref{thm: one measure full}.

\begin{theorem}[Marginal correction]		\label{thm: one measure}
  Consider a regularly varying density $f$ on the cube $\{0,1\}^p$ 
  and draw an independent sample $S$
  from the cube according to this distribution.
  If $\abs{S} \gg e^{2d} \binom{p}{\le d}$, then with probability $1-o(1)$
  there exists a density $h$ on $S$ that has exactly the same
  marginals as $f$ up to dimension $d$. Moreover, $h$ is within a $1+o(1)$ factor
  of the uniform density on $S$.
\end{theorem}

The law of large numbers tells us that the sample $S$ must have {\em approximately} the same marginals as the density $f$ from which $S$ was drawn. Theorem~\ref{thm: one measure} tells us that we can make the marginals {\em exactly} the same by a slight reweighting of $S$, i.e. by weights that are all $1+o(1)$.

\section{Fourier analysis}\label{s:fourier}

The proof of Theorem~\ref{thm: two measures} is based on hypercontractivity, duality, and empirical processes.

Let us start by recalling the basic Fourier analysis on the Boolean cube~\cite{o2014analysis}. It is more convenient to work on $\{-1,1\}^p$ than on $\{0,1\}^p$; all results are easily translatable from one cube to the other.

The {\em Walsh functions} $w_J : \{-1,1\}^p \to \{-1,1\}$ are
indexed by subsets $J \subset [p]$ and are defined as
\begin{equation}	\label{eq: Walsh}
w_J (x) = \prod_{j \in J} x(j),
\end{equation}
with the convention $w_\emptyset = 1$.

The canonical inner product on the space of real-valued functions on $\{-1,1\}^p$
is defined as
$$
\ip{f}{g}_{L^2} = \frac{1}{2^p} \sum_{x \in \{-1,1\}^p} f(x) \, g(x).
$$
This inner product defines the space $L^2 = L^2(\{-1,1\}^p)$.
More generally, for $1 \le q < \infty$, the $L^q = L^q(\{-1,1\}^p)$ is the space of
real-valued functions on the cube with the norm
$$
\norm{f}_{L^q} = \Big( \frac{1}{2^p} \sum_{x \in \{-1,1\}^p} \abs{f(x)}^q \Big)^{1/q}.
$$

Walsh functions form an orthonormal basis of $L^2$, so any function $f:\{-1,1\}^p \to \R$ admits a Fourier expansion
$$
f = \sum_{J \in [p]} \hat{f}_J w_J,
\quad \text{where } \hat{f}_J = \ip{f}{w_J}
\text{ are Fourier coefficients}.
$$
Thus, any function $f$ on the cube can be orthogonally decomposed into low and high frequencies:
$$
f = f^{\le d} + f^{>d},
$$
where
$$
f^{\le d} = \sum_{J \in [p],\, \abs{J} \le d} \ip{f}{w_J} w_J
\quad \text{and} \quad
f^{> d} = \sum_{J \in [p],\, \abs{J} > d} \ip{f}{w_J} w_J.
$$
Clearly, the function $f^{\le d}$ is determined by the Fourier coefficients of $f$ up to dimension $d$, and vice versa.

We say that a function $f$ on the cube has {\em degree at most $d$}
if $f = f^{\le d}$. Such functions form the ``low-frequency'' space
$$
W^{\le d} = \left\{ f:\; f = f^{\le d} \right\} = \Span\{w_J: \; \abs{J} \le d\},
$$
and it has dimension $\binom{p}{\le d}$.
The orthogonal complement to this subspace in $L^2$ is the
``high-frequency'' subspace
$$
W^{>d} = \left\{ f:\; f = f^{>d} \right\} = \Span\{w_J: \; \abs{J} > d\}.
$$

The following result is well known, see \cite[Theorem~9.22]{o2014analysis}:

\begin{theorem}[Hypercontractivity]		\label{thm: 1 2 hypercontractivity}
  For any $d \le p$ and any function $f: \{-1,1\}^p \to \R$ of degree at most $d$,
  we have
  $$
  \norm{f}_{L^2} \le e^d \norm{f}_{L^1}.
  $$
\end{theorem}

\subsection{Connection to marginals}

The low-degree Fourier coefficients of $f:\{-1,1\}^p \to \R$ determine the low-dimensional marginals of $f$. More precisely, $f^{\le d}$ determines the values of all marginals of $f$ up to dimension (or degree) $d$.

To see this, consider the example of the two-dimensional marginal in which the first parameter is set to $1$ and the second is set fo $-1$. The value of such marginal of $f$ is $\sum_{x \in \{-1,1\}^p} f(x) \one_{\{x(1)=1,\,x(2)=-1\}}$. Now,
$$
\one_{\{x(1)=1,\,x(2)=-1\}}(x)
= \one_{\{x(1)=1\}}(x) \one_{\{x(2)=-1\}}
= \Big( \frac{1+x(1)}{2} \Big) \Big( \frac{1-x(2)}{2} \Big),
$$
so expanding the right hand side and using the definition of Walsh functions, we see that
$$
\one_{\{x(1)=1,\,x(2)=-1\}}
= \frac{1}{4} \left( w_\emptyset + w_{\{1\}} - w_{\{2\}} - w_{\{1,2\}} \right).
$$
Thus, the marginal can be written as
$$
\sum_{x \in \{-1,1\}^p} f(x) \one_{\{x(1)=1,\,x(2)=-1\}}
= \frac{1}{4} \left( \hat{f}_\emptyset + \hat{f}_{\{1\}} - \hat{f}_{\{2\}} - \hat{f}_{\{1,2\}} \right),
$$
and so it depends only on the Fourier coefficients on $f$ up to degree $2$,
or equivalently only on $f^{\le 2}$.

\section{Empirical processes}\label{s:empirical}

Let $\mu$ be a probability measure on $\{-1,1\}^p$, and let
$$
\mu_m = \frac{1}{m} \sum_{i=1}^m \d_{\theta_i}
$$
be the corresponding (random) {\em empirical measure}, i.e., the uniform probability measure on the sample $\{\theta_1,\ldots,\theta_m\}$ of points drawn from the cube independently according to the distribution $\mu$.
These two measures define the population and empirical $L^q$ norms of functions on the cube:
\begin{equation}	\label{eq: Lq norms}
\norm{F}_{L^q(\mu)}^q \coloneqq \E \abs{F(\theta_1)}^q; \qquad
\norm{F}_{L^q(\mu_m)}^q \coloneqq \frac{1}{m} \sum_{i=1}^m \abs{F(\theta_i)}^q.
\end{equation}
We clearly have
$\E \norm{F}_{L^1(\mu_m)} = \norm{F}_{L^1(\mu)}$.
The following result provides a uniform deviation inequality.

\begin{proposition}[Deviation of the empirical $L^1$ norm]		\label{prop: sampling in L1}
  Let $\mu$ be a probability measure on $\{-1,1\}^p$
  and $\mu_m$ be the empirical counterpart.
  Then
  $$
  \E \sup_{F \in W^{\le d}, \; \norm{F}_{L^2}=1}
  \abs{\norm{F}_{L^1(\mu_m)} - \norm{F}_{L^1(\mu)}}
  \le 2 \sqrt{\frac{1}{m} \binom{p}{\le d}}.
  $$
\end{proposition}

The $L^2$ norm on the left side is with respect to the uniform probability measure on the cube.

\begin{proof}
Any function $F \in W^{\le d}$ is a linear combination of low-degree Walsh functions,
$$
F = \sum_{\abs{J} \le d} a_J w_J.
$$
Without loss of generality (by rescaling) we can assume that
\begin{equation}	\label{eq: a unit}
\norm{F}_{L^2}^2 = \sum_{\abs{J} \le d} a_J^2 = 1.
\end{equation}
By definition of the $L^1(\mu)$ norm in \eqref{eq: Lq norms}, we have
$$
\norm{F}_{L^1(\mu)}
= \E \abs{\sum_{\abs{J} \le d} a_J w_J(\theta_1)}
= \E \abs{\ip{w(\theta_1)}{a}},
$$
where, for every $\theta$ in the cube,
$w(\theta) \coloneqq \left( w_J(\theta) \right)_{\abs{J} \le d}$
is a vector in $\R^{\binom{p}{\le d}}$, and similarly $a = \left( a_J \right)_{\abs{J} \le d}$ denotes the coefficient vector in $\R^{\binom{p}{\le d}}$.
By \eqref{eq: a unit}, $a$ is a unit vector, i.e. $a \in S^{\binom{p}{\le d}-1}$.
In a similar way, the definition of the empirical $L^1$ norm in \eqref{eq: Lq norms} yields
$$
\norm{F}_{L^1(\mu_m)}
= \frac{1}{m} \sum_{i=1}^m \abs{\sum_{\abs{J} \le d} a_J w_J(\theta_i)}
= \frac{1}{m} \sum_{i=1}^m \abs{\ip{w(\theta_i)}{a}}.
$$
Then
\begin{align*}
E
  &\coloneqq \E \sup_{F \in W^{\le d}, \; \norm{F}_{L^2}=1} \abs{\norm{F}_{L^1(\mu_m)} - \norm{F}_{L^1(\mu)}} \\
  &= \E \sup_{a \in S^{\binom{p}{\le d}-1}} \abs{ \frac{1}{m} \sum_{i=1}^m \abs{\ip{w(\theta_i)}{a}} - \E \abs{\ip{w(\theta_1)}{a}} }.
\end{align*}
Applying a symmetrization inequality for empirical processes (see e.g. \cite[Exercise~8.3.24]{vershyninbook}), we get
$$
E \le 2 \E \sup_{a \in S^{\binom{p}{\le d}-1}} \abs{ \frac{1}{m} \sum_{i=1}^m \e_i \abs{\ip{w(\theta_i)}{a}} },
$$
where $(\e_i)_{i=1}^m$ denote i.i.d. Rademacher random variables, which are independent of the sample points $(\theta_i)_{i=1}^m$.

The exterior absolute value can be removed using the symmetry of the Rademacher random variables, and the interior absolute values can be removed using Talagrand's contraction principle, see \cite[Exercise~6.7.7]{vershyninbook}, thus continuing our bound as
\begin{align*}
E  &\le 2 \E \sup_{a \in S^{\binom{p}{\le d}-1}} \frac{1}{m} \sum_{i=1}^m \e_i \ip{w(\theta_i)}{a} \\
  &= 2\E \norm[3]{\frac{1}{m} \sum_{i=1}^m \e_i w(\theta_i)}_2
  \le \frac{2}{m} \Bigg( \E \norm[3]{\sum_{i=1}^m \e_i w(\theta_i)}_2^2 \Bigg)^{1/2}
  =\frac{2}{m} \Bigg( \sum_{i=1}^m \E \norm{w(\theta_i)}_2^2 \Bigg)^{1/2}
\end{align*}
where the last step follows by conditioning on $(\theta_i)$.
Since all $\binom{p}{\le d}$ coordinates of all vectors $w(\theta_i)$ equal $\pm 1$,
we have $\norm{w(\theta_i)}_2^2 = \binom{p}{\le d}$ deterministically. Substituting
this bound, we complete the proof.
\end{proof}

\section{Enforcing a uniform bound and sparsity}

We will now prove that for any function $F$ on the Boolean cube,
there is another function that simultaneously satisfies the three desiderata:
(a) it has the same marginals (or Fourier coefficients) as $F$ up to dimension $d$;
(b) it is very sparse -- in fact, it is supported on a random set of a given cardinality; and
(c) it is uniformly bounded. The following result guarantees the existence of such function $F-w$.

\begin{theorem}				\label{thm: transfer}
  Let $\mu$ be a probability measure on the cube $\{-1,1\}^p$ whose density is bounded below by $\alpha/2^p$,
  and let $\mu_m$ be the empirical counterpart.
  If $m \ge 16 (\a\g)^{-2} e^{2d} \binom{p}{\le d}$, then the following holds with
  probability at least $1-\g$. For any function $F: \{-1,1\}^p \to \R$, we have
  $$
  \inf \left\{ \norm{F-w}_\infty :\; w \in W^{>d}, \, F-w \subset S_{\mu_m} \right\}
  \le \frac{2e^d 2^p}{\alpha m} \norm{F^{\le d}}_{L^2}
  $$
  where $S_{\mu_m}$ denotes the set of the functions supported on $\supp(\mu_m)$.\end{theorem}

Throughout the proof, let us denote
$$
S \coloneqq \supp(\mu_m).
$$
The $L^1$ norm of any function $F:\{-1,1\}^p \to \R$ naturally decomposes as
$$
\norm{F}_{L^1} = \norm{F \one_S}_{L^1} + \norm{F \one_{S^c}}_{L^1},
$$
where $\one_S$ denotes the indicator function of $S$.
Given $\d>0$, consider the weighted space $L^1_\d$ where the norm is defined by
$$
\norm{F}_{L^1_\d}
\coloneqq \norm{F \one_S}_{L^1} + \d \norm{F \one_{S^c}}_{L^1}.
$$

\begin{lemma}		\label{lem: 1weighted 2}
  Consider the subspace $(W^{\leq d},\|\,\|_{L_{\delta}^{1}})$ of $L_{\delta}^{1}$. With probability at least $1-\g$, for every $\d>0$ we have
  $$
  \norm{\Id:\; (W^{\leq d},\|\,\|_{L_{\delta}^{1}}) \to L^2}
  \le \frac{2e^d 2^p}{\alpha m}.
  $$
\end{lemma}

\begin{proof}
Proposition~\ref{prop: sampling in L1} combined with Markov's inequality and rescaling
implies that, with probability $1-\g$, the following holds for all $F \in W^{\le d}$:
$$
\abs{\norm{F}_{L^1(\mu)} - \norm{F}_{L^1(\mu_m)}}
\le \frac{2}{\gamma} \sqrt{\frac{1}{m} \binom{p}{\le d}} \; \norm{F}_{L^2}
\le \frac{\alpha}{2 e^d} \norm{F}_{L^2},
$$
where in the last step we used the assumption on $m$.

Applying hypercontractivity (Theorem~\ref{thm: 1 2 hypercontractivity}), the regularity assumption of $\mu$, and the bound above, we obtain
$$
\frac{1}{e^d} \norm{F}_{L^2}
\le \norm{F}_{L^1}
\le \frac{1}{\a} \norm{F}_{L^1(\mu)}
\le \frac{1}{\a} \norm{F}_{L^1(\mu_m)} + \frac{1}{2 e^d} \norm{F}_{L^2}.
$$
Rearranging the terms, we obtain
$$
\frac{1}{2 e^d} \norm{F}_{L^2}
\le \frac{1}{\a} \norm{F}_{L^1(\mu_m)}
= \frac{2^p}{\a m} \norm{F \one_S}_{L^1}
\le \frac{2^p}{\a m} \norm{F}_{L^1_\d}
$$
where in the middle step we used the definitions of $S$ and of the norms in $L^1(\mu)$ and  $L^1(\mu_m)$.
Multiplying both sides by $2e^d$ completes the proof.
\end{proof}

\medskip

\begin{proof}[Proof of Theorem~\ref{thm: transfer}]
Let us dualize Lemma~\ref{lem: 1weighted 2} with respect to the inner product on $L^2$. The identity operator is self-adjoint, and the adjoint operator has the same norm. So, with  probability at least $1-\g$, for every $\d>0$ we have
$$
\norm{\Id:\; \big( L^2 \big)^* \to \big( W^{\leq d},\|\,\|_{L_{\delta}^{1}} \big)^*}
\le \frac{2e^d 2^p}{\alpha m}
\eqqcolon B.
$$
The Hilbert space $L^2$ is self-dual.
The dual to the weighted space $L^1_\d$ is the weighted space $L^\infty_{1/\d}$ defined as
\begin{equation}	\label{eq: Linfty weighted}
\norm{F}_{L^\infty_{1/\d}}
\coloneqq \norm{F \one_S}_{L^\infty} \vee \frac{1}{\d} \norm{F \one_{S^c}}_{L^\infty}.
\end{equation}
The dual of a subspace is a quotient space of the dual:
$$
\big(W^{\leq d},\|\,\|_{L_{\delta}^{1}} \big)^*
= \big( L^1_\d \big)^* /(W^{\le d})^\perp = L^\infty_\d / W^{>d}.
$$
Putting these considerations together, we get
$$
\norm{\Id:\; L^2 \to L^\infty_\d / W^{> d}}
\le B.
$$

By definition of the quotient norm, this bound means that for every function $F:\{-1,1\}^p \to \R$ there exists $w \in W^{>d}$ such that
$$
\norm{F-w}_{L^\infty_\d} \le B \norm{F}_{L^2}.
$$
By definition \eqref{eq: Linfty weighted} of the weighted norm, this means that
\begin{equation}	\label{eq: F-w weighted}
\norm{(F-w) \one_S}_\infty \le B \norm{F}_{L^2}
\quad \text{and} \quad
\norm{(F-w) \one_{S^c}}_\infty \le \d B \norm{F}_{L^2}.
\end{equation}
Since the second bound holds for arbitrary $\d>0$, it follows that $\norm{(F-w) \one_{S^c}}_\infty = 0$, i.e.
$$
\supp(F-w) \subset S
$$
as claimed in the theorem. Together with the first bound in \eqref{eq: F-w weighted}, this proves that
$$
\norm{F-w}_\infty \le B \norm{F}_{L^2}.
$$
Thus, we showed every function $F:\{-1,1\}^p \to \R$ satisfies
$$
\inf \left\{ \norm{F-w}_\infty :\; w \in W^{>d}, \, F-w \subset S_{\mu_m} \right\}
\le B \norm{F}_{L^2}
$$
Finally, note that the term $\norm{F}_{L^2}$ on the right hand side
can automatically be improved to $\norm[1]{F^{\le d}}_{L^2}$.
To see this, apply the above bound for $F^{\le d}$
and absorb the term $F^{>d}$ into $w$.
Theorem~\ref{thm: transfer} is proved.
\end{proof}

\section{Low-degree projections of empirical measures}\label{s:lowdegree}

Consider two probability measures $\nu$ and $\mu$ on $\{-1,1\}^p$,
and let $f$ and $g$ denote their densities (or probability mass functions):
$$
f(z) = \nu(\{z\})
\quad \text{and} \quad
g(z) = \mu(\{z\}),
\quad z \in \{-1,1\}^p.
$$
The densities of the empirical probability measures $\nu_n$ and $\mu_m$ are
\begin{equation}	\label{eq: fn gn}
f_n = \frac{1}{n} \sum_{i=1}^n \one_{x_i}
\quad \text{and} \quad
g_m = \frac{1}{m} \sum_{i=1}^m \one_{y_i}
\end{equation}
where $x_1,\ldots,x_n$ and $y_1,\ldots,y_m$ are i.i.d.\ points drawn from the cube according to the densities $f$ and $g$, respectively.
The functions $f_n$ and $g_m$ provide unbiased estimators of $f$ and $g$:
$$
\E f_n = f, \quad \E g_m = g.
$$

Assume that $f(z)=0$ whenever $g(z)=0$. Consider the function
\begin{equation}	\label{eq: gn tilde}
\tilde{g}_m \coloneqq (f/g) g_m.
\end{equation}
Although $\tilde{g}_m$ is supported on the sample drawn from density $g$,
it provides an unbiased estimator of $f$:
$$
\E \tilde{g}_m = (f/g) \E g_m = (f/g)g = f.
$$
This property will be crucial in the proof of Theorem~\ref{thm: two measures}.

Let us look at the low-degree projections of $f_n$ and $\tilde{g}_m$ and try to bound their mean magnitude and deviation from the mean. Toward this end, note that
\begin{equation}	\label{eq: 1x norm}
\forall x \in \{-1,1\}^p, \quad
\norm[1]{(\one_x)^{\le d}}_{L^2} = \binom{p}{\le d}^{1/2} \frac{1}{2^p}.
\end{equation}
Indeed, to see this, use Parseval's identity
$$
\norm[1]{(\one_x)^{\le d}}_{L^2}^2
= \sum_{\abs{J} \le d} \ip{\one_x}{w_J}_{L^2}^2
= \sum_{\abs{J} \le d} \Big( \frac{1}{2^p} w_J(x) \Big)^2
$$
and recall that the Walsh function $w_J$ takes $\pm 1$ values.
Furthermore, by definition of $f_n$ and the triangle inequality, \eqref{eq: 1x norm} yields
\begin{equation}	\label{eq: fn norm}
\norm[1]{(f_n)^{\le d}}_{L^2} \le \binom{p}{\le d}^{1/2} \frac{1}{2^p}
\quad \text{deterministically}.
\end{equation}

\begin{lemma}[Deviation]	\label{lem: deviation fn gm}
  We have
  $$
  \left( \E \norm[1]{(f_n-f)^{\le d}}_{L^2}^2 \right)^{1/2}
  \le \binom{p}{\le d}^{1/2} \frac{1}{\sqrt{n} 2^p}.
  $$
  Moreover, if $\norm{f/g}_{L^2} \le \kappa$ then we have
  $$
  \left( \E \norm[1]{(\tilde{g}_m-f)^{\le d}}_{L^2} \right)^{1/2}
  \le \binom{p}{\le d}^{1/2} \frac{\kappa}{\sqrt{m} 2^p}.
  $$
\end{lemma}

\begin{proof}
By Parseval's identity,
\begin{equation}	\label{eq: Parseval for deviation}
\norm[1]{(f_n-f)^{\le d}}_{L^2}^2
= \sum_{\abs{J} \le d} \ip{f_n-f}{w_J}_{L^2}^2.
\end{equation}
By definition \eqref{eq: fn gn} of $f_n$, each term of this sum can be expressed as
$$
\ip{f_n-f}{w_J}_{L^2} = \frac{1}{n} \sum_{i=1}^n \ip{\one_{x_i}-f}{w_J}_{L^2}.
$$
The terms on the right hand side are i.i.d. mean zero random variables, so
\begin{align*}
\E \ip{f_n-f}{w_J}_{L^2}^2
  &= \frac{1}{n} \E \ip{\one_{x_1}-f}{w_J}_{L^2}^2 \\
  &\le \frac{1}{n} \E \ip{\one_{x_1}}{w_J}_{L^2}^2
  	\quad \text{(the variance is bounded by the second moment)} \\
  &= \frac{1}{n} \E \Big( \frac{1}{2^p} w_J(x_1) \Big)^2
  = \frac{1}{n 2^{2p}},
\end{align*}
since the Walsh function $w_J$ takes $\pm 1$ values.
Substitute this bound into Parseval's identity \eqref{eq: Parseval for deviation} to get
$$
\E \norm[1]{(f_n-f)^{\le d}}_{L^2}^2
\le \binom{p}{\le d} \cdot \frac{1}{n 2^{2p}}.
$$
This proves the first part of the lemma.

The second part of the lemma can be derived similarly. Indeed,
\begin{equation}	\label{eq: Parseval for deviation g}
\norm[1]{(\tilde{g}_m-f)^{\le d}}_{L^2}^2
= \sum_{\abs{J} \le d} \ip{\tilde{g}_m-f}{w_J}_{L^2}^2.
\end{equation}
By definition \eqref{eq: fn gn} of $g_m$ and \eqref{eq: gn tilde} of $\tilde{g}_m$,
each term of this sum can be expressed as
$$
\ip{\tilde{g}_m-f}{w_J}_{L^2} = \frac{1}{m} \sum_{i=1}^m \Bigip{\frac{f(y_i)}{g(y_i)} \cdot \one_{y_i}-f}{w_J}_{L^2}.
$$
The terms on the right hand side are i.i.d. mean zero random variables, so
\begin{align*}
\E \ip{\tilde{g}_m-f}{w_J}_{L^2}^2
  &= \frac{1}{m} \E \Bigip{\frac{f(y_1)}{g(y_1)} \cdot \one_{y_1}-f}{w_J}_{L^2}^2 \\
  &\le \frac{1}{m} \E \Bigip{\frac{f(y_1)}{g(y_1)} \cdot \one_{y_1}}{w_J}_{L^2}^2
  	\quad \text{(the variance is bounded by the second moment)} \\
  &= \frac{1}{m} \E \Big( \frac{1}{2^p} \frac{f(y_1)}{g(y_1)} w_J(y_1) \Big)^2 \\
  &= \frac{1}{m 2^{2p}} \norm{f/g}_{L^2}^2
  \le \frac{\kappa^2}{m 2^{2p}},
\end{align*}
where in the last line we used the fact that the Walsh function $w_J$ takes $\pm 1$ values
and the assumption on $f/g$.
Substitute this bound into Parseval's identity \eqref{eq: Parseval for deviation g} to get
$$
\E \norm[1]{(\tilde{g}_m-f)^{\le d}}_{L^2}^2
\le \binom{p}{\le d} \cdot \frac{\kappa^2}{m 2^{2p}}.
$$
This proves the second part of the lemma.
\end{proof}

\section{Proof of Theorem~\ref{thm: two measures}}

The following master theorem is a more general version of Theorem~\ref{thm: two measures}, as we will see shortly. Recall that $g_m,\mu_m,\tilde{g}_m$ are defined in (\ref{eq: fn gn}).

\begin{theorem}		\label{thm: master}
  Let $f$ and $g$ be densities on the cube $\{-1,1\}^p$,
  and let $f_n$ and $g_m$ be their empirical counterparts.
  Assume that $\norm{f/g}_{L^2} \le \kappa$ for some $\kappa \ge 1$ and
  that $g$ is bounded below by $\alpha/2^p$.
  If $n \ge 16 (\a\d)^{-2}\g^{-1} e^{2d} \binom{p}{\le d}$ and
  $m \ge 16 (\a\d)^{-2}\g^{-1} \kappa^{2} e^{2d} \binom{p}{\le d}$
  then the following holds with probability $1-2\gamma$.
  There exists $h: \{-1,1\}^p \to \R$ that satisfies
  $$
  h^{\le d} = f_n^{\le d},
  \quad
  \supp(h) \subset \supp(g_m),
  \quad
  \norm{h - (f/g) g_m}_\infty \le \frac{\d}{m}.
  $$
\end{theorem}

\begin{proof}
Let $\tilde{g}_m = (f/g) g_m$ and apply Theorem~\ref{thm: transfer} for the function
$F = f_n-\tilde{g}_m$. With probability $1-\g$, there exists $w \in W^{>d}$ such that
\begin{equation}	\label{eq: fn gm w}
f_n-\tilde{g}_m-w \in S_{\mu_m}
\quad \text{and} \quad
\norm{f_n-\tilde{g}_m-w}_\infty
\le \frac{2e^d 2^p}{\alpha m} \norm[1]{(f_n-\tilde{g}_m)^{\le d}}_{L^2}.
\end{equation}
Set
$$
h = f_n-w.
$$
Since $w \in W^{>d}$, we have $h^{\le d} = f_n^{\le d}$ as claimed.
Since both $\tilde{g}_m$ and $h-\tilde{g}_m = f_n-\tilde{g}_m-w$ lie in $S_{\mu_m}$, so does $h$, as claimed.

Furthermore, combining both bounds of Lemma~\ref{lem: deviation fn gm} via the Minkowski inequality, we get
$$
\left( \E \norm[1]{(f_n - \tilde{g}_m)^{\le d}}_{L^2}^2 \right)^{1/2}
\le \binom{p}{\le d}^{1/2} \Big( \frac{1}{\sqrt{n}} + \frac{\kappa}{\sqrt{m}} \Big) \frac{1}{2^p}.
$$
By Chebyshev's inequality, with probability at least $1-\g$ we have
$$
\norm[1]{(f_n - \tilde{g}_m)^{\le d}}_{L^2}
\le \g^{-1/2} \binom{p}{\le d}^{1/2} \Big( \frac{1}{\sqrt{n}} + \frac{\kappa}{\sqrt{m}} \Big) \frac{1}{2^p}.
$$
We substitute this into~\eqref{eq: fn gm w} and get
$$
\norm{h - \tilde{g}_m}_{L^\infty(\nu_m)}
\le \frac{2e^d 2^p}{\alpha m} \cdot \g^{-1/2} \binom{p}{\le d}^{1/2} \Big( \frac{1}{\sqrt{n}} + \frac{\kappa}{\sqrt{m}} \Big) \frac{1}{2^p}
\le \frac{\d}{m},
$$
where we used the assumption on $n$ and $m$ in the last bound.
\end{proof}

\subsection{Proof of Theorem~\ref{thm: two measures}}

Let us explain how Theorem~\ref{thm: master} is a more general form of Theorem~\ref{thm: two measures}.
Let $f$ and $g$ be the densities of the two distributions in the statement of Theorem~\ref{thm: two measures},
$X = (x_1,\ldots,x_n)$ and $S = (y_1,\ldots,y_m)$
be the samples drawn according to these densities,
and $f_n = \frac{1}{n} \sum_{i=1}^n \one_{x_i}$ and
$g_m = \frac{1}{m} \sum_{i=1}^m \one_{y_i}$ be the empirical densities.
The regularity assumption implies that
\begin{equation}	\label{eq: f/g constant}
f/g \asymp 1 \quad \text{pointwise},
\end{equation}
and in particular the requirement $\norm{f/g}_{L^2} = O(1)$ holds in Theorem~\ref{thm: master}. The function $h$ we obtain from that result is supported on
$S = \supp(g_m)$ and satisfies
$$
h
\ge (f/g)g_m - \frac{\d}{m}
\gtrsim \frac{1}{m}
\quad \text{everywhere on $S$}.
$$
(In the last step we used \eqref{eq: f/g constant} that $g_m = \frac{1}{m} \sum_{i=1}^m \one_{y_i}$ is lower bounded by $1/m$ on $S$.)
In particular, $h$ is positive on $S$. The condition $h^{\le d} = f_n^{\le d}$ means that $h$ has exactly the same marginals up to dimension $d$ as $f_n$, the uniform probability distribution on $X$. Since $f_n$ is a density, the sum of all of its values equals $1$. The same must be true for $h$, since the sum of the values can be expressed as the zero-dimensional marginal, which must be the same for $h$ and $f_n$. In other words, $h$ must be a density, too. Theorem~\ref{thm: two measures} is proved. \qed

\subsection{A one-sample version}

Here is a one-sample version of Theorem~\ref{thm: master}.
It is a rigorous version of Theorem~\ref{thm: one measure} we stated informally in the introduction.

\begin{theorem}		\label{thm: one measure full}
  Let $f$ be a density on the cube $\{-1,1\}^p$ that is bounded below by $\alpha/2^p$,
  and let $f_m$ be its empirical counterpart.
  If $m \ge 16 (\a\d)^{-2}\g^{-1} e^{2d} \binom{p}{\le d}$
  then the following holds with probability $1-2\gamma$.
  There exists a density $h$ on $\supp(f_m)$ that satisfies
  $$
  h^{\le d} = f^{\le d},
  \quad
  \norm{h - f_m}_\infty \le \frac{\d}{m}.
  $$
\end{theorem}

\begin{proof}
The proof is similar to that of Theorem~\ref{thm: two measures} above.
Choose $g=f$, $n=m$, hence $\tilde{g}_m = (f/g)g_m = f_m$,
and use $F = f-\tilde{g}_m$.
Apply only the first bound in Lemma~\ref{lem: deviation fn gm}.

Note that the bound in the conclusion and the fact that $f_m=1/m$ on its support implies that $h \ge 1/m-\d/m > 0$ on $\supp(f_m)$, and thus $h$ is a density.

We leave the details to the reader.
\end{proof}

\section{Solution space}\label{s:solutionspace}

Our next focus is on proving Theorem~\ref{thm: privacy-accuracy},
which gives guarantees for privacy and accuracy of the synthetic data created by  Algorithm~\ref{mainalgorithm1}.

Let us formally introduce the solution space -- the space of all functions on the reduced sample space $S$ that have the same marginals as a given function $u$.

\begin{definition}[Solution space]
  Let $\mu$ be a probability measure on the cube $\{-1,1\}^p$,
  and $\mu_m$ be its empirical counterpart.
  For any function $u: \{-1,1\}^p \to \R$, consider the affine subspace $H(u)$ of all functions supported on $\supp(\mu_m)$ and that have the same marginals up to dimension $d$ as the function $u$, i.e.
  $$
  H(u) \coloneqq \left\{ h \in S_{\mu_m} :\; h^{\le d} = u^{\le d} \right\}
  = \left( u-W^{>d} \right) \cap S_{\mu_m},
  $$
  where $S_{\mu_m}$, as before, denotes the linear space of all functions supported on
  the reduced space $S = \supp(\mu_m)$.
\end{definition}

\subsection{Success with high probability}

The Algorithm~\ref{mainalgorithm1} succeeds, i.e. does not return ``Failure'', when the reduced space $S = \{\theta_1,\ldots,\theta_m\}$ is well conditioned. By definition, this happens if
\begin{equation}	\label{eq: smin}
s_{\min}(M) \ge \frac{\sqrt{m}}{2e^d}
\end{equation}
where $s_{\min}$ denotes the smallest singular value, and $M$ is the $m \times \binom{p}{\le d}$ matrix whose entries are $w_J(\theta_i)$ for $\abs{J} \le d$, i.e. the matrix whose rows are indexed by the points $\theta_i \in S$, and whose columns are indexed by Walsh functions $w_J$ of degree at most $d$.

Let us reformulate the condition \eqref{eq: smin} in the dual form,
and then deduce from Theorem~\ref{thm: transfer} that that it holds with high probability.

\begin{lemma}[Well conditioned reduced space]		\label{lem: well conditioned}
  The reduced space $S$ is well conditioned if and only if
  any function $F: \{-1,1\}^p \to \R$ satisfies
  \begin{equation}	\label{eq: well conditioned}
  \inf \left\{ \norm{F-w}_{L^2(\mu_m)} :\; w \in W^{>d}, \, F-w \in S_{\mu_m} \right\}
  \le \frac{2e^d 2^p}{m} \norm{F^{\le d}}_{L^2}.
  \end{equation}
\end{lemma}

\begin{proof}
Decomposing $F = F^{\le d} + F^{>d}$ we see that $F^{\le d}$ in the right hand side of \eqref{eq: well conditioned} may be replaced by $F$ without loss of generality.
Furthermore, since $\norm{f}_{L^2(\mu_m)} = \sqrt{2^p/m}\norm{f}_{L^2}$ for any $f \in S_{\mu_m}$, we can rewrite condition \eqref{eq: well conditioned} equivalently as
\begin{equation}	\label{eq: well conditioned L2}
\inf \left\{ \norm{F-w}_{L^2} :\; w \in W^{>d}, \, F-w \in S_{\mu_m} \right\}
\le B \norm{F}_{L^2}
\end{equation}
where
$$
B = 2e^d \sqrt{\frac{2^p}{m}}.
$$

We will employ a duality argument similar to the one we used in the proof of Theorem~\ref{thm: transfer}.
Given $\d>0$, consider the weighted Hilbert space $L^2_\d$ where the norm is defined by
$$
\norm{F}_{L^2_\d}^2
\coloneqq \norm{F \one_S}_{L^2}^2 + \d \norm{F \one_{S^c}}_{L^2}^2.
$$
where $\one_S$ denotes the indicator function of $S$.
Then \eqref{eq: well conditioned L2} is equivalent to
$$
\inf \left\{ \norm{F-w}_{L^2_{1/\d}} :\; w \in W^{>d} \right\}
\le B \norm{F}_{L^2}
\quad \forall \d>0.
$$
(To see this, note that taking $\d \to 0_+$ enforces $F-w \one_{S^c} = 0$, or equivalently $F-w \in S_{\mu_m}$.)
This can be interpreted as a bound on the norm of the quotient map $Q$:
$$
\norm{Q:\; L^2 \to L^2_{1/\d}/W^{>d}} \le B
\quad \forall \d>0.
$$
Let us dualize this bound. The adjoint operator has the same norm, so
$$
\norm{Q^*:\; \big( L^2 \big)^* \to \big( L^2_{1/\d}/W^{>d} \big)^*} \le B
\quad \forall \d>0.
$$
The adjoint of the quotient map is the canonical (identity) embedding;
the Hilbert space $L^2$ is self-dual, and the dual of a quotient space is a subspace of the dual, i.e.
$$
\big( L^2_{1/\d}/W^{>d} \big)^*
= ((W^{>d})^\perp,\|\,\|_{( L^2_{1/\d})^*})
= (W^{\le d},\|\,\|_{L^2_\d}).
$$
Thus, the bound is equivalent to
$$
\norm{\Id:\; (W^{\le d},\|\,\|_{L^2_\d}) \to L^2} \le B
\quad \forall \d>0.
$$

By definition of the operator norm and the norm in $L^2_\d$, this bound is equivalent to
saying that
$$
\norm{F}_{L^2}^2 \le B^2 \left( \norm{F \one_S}_{L^2}^2 + \d \norm{F \one_{S^c}}_{L^2}^2 \right)
\quad \forall F \in W^{\le d}, \, \forall \d>0.
$$
Taking $\d \to 0_+$, we see that this is equivalent to
$$
\norm{F}_{L^2}^2 \le B^2 \norm{F \one_S}_{L^2}^2
= \frac{B^2}{2^p} \norm{F \one_S}_{\ell^2}^2
= \frac{4e^{2d}}{m} \norm{F \one_S}_{\ell^2}^2
\quad \forall F \in W^{\le d}.
$$
Expressing $F$ through its orthogonal decomposition
$F = \sum_{\abs{J} \le d} a_J w_J$, we can rewrite the latter condition as
$$
\sum_{\abs{J} \le d} a_J^2
\le \frac{4e^{2d}}{m} \norm[3]{\sum_{\abs{J} \le d} a_J w_J \one_S}_{\ell^2}^2
= \frac{4e^{2d}}{m} \sum_{i=1}^m \Big( \sum_{\abs{J} \le d} a_J w_J(\theta_i) \Big)^2
\quad \forall \text{ choice of coefficients } a_J.
$$
This in turn is equivalent to
$$
\norm{a}_{\ell^2}^2 \le \frac{4e^{2d}}{m} \norm{Ma}_{\ell^2}^2,
$$
which is finally equivalent to \eqref{eq: smin}.
\end{proof}

\begin{proposition}[Success with high probability]	\label{prop: success}
  If $m \ge 16 \g^{-2} e^{2d} \binom{p}{\le d}$, then Algorithm~\ref{mainalgorithm1} succeeds (i.e. does not return ``Failure'') with probability at least $1-\gamma$.
\end{proposition}

\begin{proof}
By definition, Algorithm~\ref{mainalgorithm1} succeeds if the reduced space $S$ is well conditioned. Then the conclusion
immediately follows from Theorem~\ref{thm: transfer} for the uniform density $\mu$, Lemma~\ref{lem: well conditioned} and the fact that the $L^2(\mu_m)$ norm is bounded by the sup-norm.
\end{proof}

\subsection{All solution spaces are translates of each other}

First let us show that with high probability in $\mu_m$,
all solution spaces $H(u)$ are nonempty and are translates of each other.
The following elementary lemma will help us.

\begin{proposition}		\label{prop: translates of each other}
  If the reduced space $S$ is well conditioned,
  the solution spaces $H(u)$ for all $u:\{-1,1\}^p \to \R$
  are nonempty and are translates of each other.
\end{proposition}

\begin{proof}
Let $F: \{-1,1\}^p \to \R$ be an arbitrary function. If $S$ is well conditioned,
Lemma~\ref{lem: well conditioned} for $F=u$
yields the existence of $w \in W^{>d}$ and $s \in S_{\mu_m}$ such that
$u=s+w$. This implies that $u-W^{>d}=s-W^{>d}$. Hence
$$
H(u)
= \left( u-W^{>d} \right) \cap S_{\mu_m}
= \left( s-W^{>d} \right) \cap S_{\mu_m}
= s-\left(W^{>d} \cap S_{\mu_m} \right).
$$
The linear subspace $W^{>d} \cap S_{\mu_m}$ is nonempty as it contains the origin.
Therefore, all solution spaces $H(u)$ are translates of this linear space, and thus of each other.
\end{proof}

\subsection{Sensitivity of the solution space}

Next, we will check that the map $u \mapsto H(u)$ is Lipschitz
in the Hausdorff metric. Recall that the Hausdorff distance between two subsets $A$ and $B$ of a normed space $X$ is defined as
$$
d_X(A,B) = \max \left\{ \sup_{a \in A} \inf_{b \in B} \norm{a-b}_X, \, \sup_{b \in B} \inf_{a \in A} \norm{a-b}_X \right\}.
$$
When $A$ and $B$ are affine subspaces that are translates of each other, we have
$$
d_X(A,B) = \inf_{b \in B} \norm{a-b}_X
= \dist_X(a,B)
\quad \text{for any } a \in A.
$$
When the norm is clear from the context, we skip the subscript $X$. When $X = L^q$ we simply write $d_q(A,B)$.

\begin{lemma}[Sensitivity of the solution space]				\label{lem: sensitivity of solution space}
  If the reduced space $S$ is well conditioned,
  then any pair of functions $u_1, u_2: \{-1,1\}^p \to \R$ satisfies
  \begin{equation}	\label{eq: sensitivity of solution space}
  d_\infty \left( H(u_1), \, H(u_2) \right)
  \le \frac{2e^d2^p}{\sqrt{m}} \, \norm[1]{(u_1-u_2)^{\le d}}_{L^2}.
  \end{equation}
\end{lemma}

\begin{proof}
Since, by Proposition~\ref{prop: translates of each other},
the affine subspaces $H(u_1)$ and $H(u_2)$ are translates of each other,
it suffices to bound $\inf_{s_2 \in H(u_2)}\norm{s_1-s_2}_\infty$ for any $s_1 \in H(u_1)$.

Pick any $s_1 \in H(u_1)$. Since $H(u_1) = (u_1-W^{>d}) \cap S_{\mu_m}$, there exists $w_1 \in W^{>d}$ such that
$s_1=u_1-w_1 \in S_{\mu_m}$.
Apply the bound in Lemma~\ref{lem: well conditioned} for $F=s_1-u_2$.
There exists $w_2 \in W^{>d}$ such that $s_1-u_2-w_2 \in S_{\mu_m}$ and
\begin{equation}	\label{eq: s1-s2}
\norm{s_1-u_2-w_2}_\infty
\le \sqrt{m} \norm{s_1-u_2-w_2}_{L^2(\mu_m)}
\le \frac{2e^d2^p}{\sqrt{m}} \, \norm[1]{(s_1-u_2-w_2)^{\le d}}_{L^2}.
\end{equation}
Since both $s_1$ and $s_1-u_2-w_2$ lie in the linear subspace $S_{\mu_m}$, it must be that $s_2 \coloneqq u_2+w_2 \in S_{\mu_m}$ as well.
Since $w_2 \in W^{>d}$, it follows that $s_2 \in (u_2+W^{>d}) \cap S_{\mu_m} = H(u_2)$.

Furthermore,
$$
(s_1-u_2-w_2)^{\le d} = (u_1-w_1-u_2-w_2)^{\le d} = (u_1-u_2)^{\le d}.
$$
(In the last step, we used that $w_1$ and $w_2$ are in $W^{>d}$ and so $(w_1)^{\le d} = (w_2)^{\le d} = 0$.)

Therefore, we can rewrite \eqref{eq: s1-s2} as
$$
\norm{s_1-s_2}_\infty \le \frac{2e^d2^p}{\sqrt{m}} \, \norm[1]{(u_1-u_2)^{\le d}}_{L^2}.
$$
The proof is complete.
\end{proof}

\subsection{Changing a single data point}

The Sensitivity Lemma~\ref{lem: sensitivity of solution space} will be applied in the situation where $u_1$ and $u_2$ are the uniform densities on the two datasets $X_1$ and $X_2$ that are different by a single element. Let us specialize the bound \eqref{eq: sensitivity of solution space} to this case.

Suppose $X_1 = (x_1,\ldots,x_n)$ and $X_2 = (x_1,\ldots,x_n,x_{n+1})$. Here, in our discussion of privacy, we allow $x_i$ be arbitrary points drawn from $\{-1,1\}^p$; they do not need to be random. The corresponding densities are
$$
f_n = \frac{1}{n} \sum_{i=1}^n \one_{x_i}
\quad \text{and} \quad
f_{n+1} = \frac{1}{n+1} \sum_{i=1}^{n+1} \one_{x_i}.
$$

A direct calculation yields
$$
f_{n+1} - f_n = \frac{1}{n+1} \left( \one_{x_{n+1}}-f_n \right).
$$
Using triangle inequality and then \eqref{eq: 1x norm} and \eqref{eq: fn norm}, we get
\begin{equation}	\label{eq: fn+1 fn}
\norm[1]{(f_{n+1} - f_n)^{\le d}}_{L^2}
\le \frac{1}{n+1} \Big( \norm[1]{(\one_{x_{n+1}})^{\le d}}_{L^2} + \norm[1]{(f_n)^{\le d}}_{L^2}\Big)
\le \frac{2}{n} \binom{p}{\le d}^{1/2} \frac{1}{2^p}.
\end{equation}

\section{Selection rule}

Next, we want to extend sensitivity to the selection rule. Can we pick one point from a solution space in such a way that a small change in the solution space always leads to a small change in the selected point?

\subsection{$L^{2}$ sensitivity}

We do not know the best selection rule in the $L^{\infty}$ metric. The problem is simpler for the $L^2$ metric: the proximal point (to a given reference point) is a good selection rule.

\begin{lemma}[Sensitivity of the closest point in the Hilbert space]		\label{lem:proximalinL2}
  Consider a Hilbert space $X$ and a reference point $r \in X$.
  Let $x(K)$ denote a point in a nonempty closed set $K \subset X$ that is
  closest to $r$, i.e.
  $$
  x_r(K) = \argmin \left\{ \norm{x-r}:\; x \in K \right\}.
  $$
  Then, for any two nonempty closed convex sets $K_1, K_2 \subset X$, we have
  $$
  \norm{x_r(K_1)-x_r(K_2)}^2 \le 4\max \left( \dist(r,K_1), \dist(r,K_2) \right)
  	\cdot d(K_1,K_2).
  $$
\end{lemma}

In order to prove this lemma, we first observe:

\begin{lemma}\label{1}
  Suppose that $K$ is a nonempty closed convex subset of a Hilbert space $X$.
  Let $r \in X$.
  Let $x_0 = \argmin \left\{ \|x-r\|: \; x\in K \right\}$. Then
  $$
  \|x_0-y\|^{2} \le 2 \left( \|y-r\|^{2}-\|x_{0}-r\|^{2} \right)
  $$
  for all $y\in K$.
\end{lemma}
\begin{proof}
Without loss of generality, assume that $r=0$. Let $y\in K$. Since $\frac{x_{0}+y}{2}\in K$, we have $\left\|\frac{x_{0}+y}{2}\right\|\geq\|x_{0}\|$, so
\[\left\|\frac{x_{0}-y}{2}\right\|^{2}+\|x_{0}\|^{2}\leq\left\|\frac{x_{0}-y}{2}\right\|^{2}+\left\|\frac{x_{0}+y}{2}\right\|^{2}=\frac{1}{2}(\|x_{0}\|^{2}+\|y\|^{2}).\]
Thus, $\|x_{0}-y\|^{2}\leq 2(\|y\|^{2}-\|x_{0}\|^{2})$.
\end{proof}

\begin{proof}[Proof of Lemma~\ref{lem:proximalinL2}]
If $d(K_{1},K_{2})\geq d(r,K_{1})+d(r,K_{2})$, then we are done, since
\begin{eqnarray*}
\|x_{r}(K_{1})-x_{r}(K_{2})\|&\leq&\|x_{r}(K_{1})-r\|+\|x_{r}(K_{2})-r\|\\&=&
d(r,K_{1})+d(r,K_{2})\leq\sqrt{(d(r,K_{1})+d(r,K_{2}))d(K_{1},K_{2})}.
\end{eqnarray*}
Thus, we may assume that $d(K_{1},K_{2})\leq d(r,K_{1})+d(r,K_{2})$. Without loss of generality, we may also assume that $d(r,K_{2})\leq d(r,K_{1})$. By Lemma \ref{1},
\[\|x_{r}(K_{1})-y\|^{2}\leq 2(\|y-r\|^{2}-d(r,K_{1})^{2}),\]
for all $y\in K_{1}$. Note that we can write $x_{r}(K_{2})=y+d(K_{1},K_{2})z$ for some $y\in K_{1}$ and $z\in X$ with $\|z\|\leq 1$. Since
\[\|y-r\|\leq\|x_{r}(K_{2})-r\|+d(K_{1},K_{2})=d(r,K_{2})+d(K_{1},K_{2}),\]
it follows that
\begin{align*}
&\|x_{r}(K_{1})-y\|^{2}\\\leq&
2[(d(r,K_{2})+d(K_{1},K_{2}))^{2}-d(r,K_{1})^{2}]\\=&
2[d(r,K_{2})+d(K_{1},K_{2})+d(r,K_{1})][d(r,K_{2})+d(K_{1},K_{2})-d(r,K_{1})]\\\leq&
2[d(r,K_{2})+d(K_{1},K_{2})+d(r,K_{1})]d(K_{1},K_{2})\\\leq&
4(d(r,K_{1})+d(r,K_{2}))d(K_{1},K_{2}),
\end{align*}
where the second inequality follows from the assumption that $d(r,K_{2})\leq d(r,K_{1})$ and the last inequality follows from the assumption that $d(K_{1},K_{2})\leq d(r,K_{1})+d(r,K_{2})$.
\end{proof}

\subsection{Restriction onto the cube}

Functions that comprise the solution space $H(u)$ may take negative values, hence not all of $H(u)$ consists of densities. So, our next goal is to restrict the affine space $H(u)$ to the positive orthant $[0,\infty)^m$ and check that sensitivity still holds. Our Algorithm~\ref{mainalgorithm1} makes a more aggressive restriction onto the cube $[2\d/m, (\Delta-\d)/m]^m$.  This is what we will analyze now.

\begin{lemma}[Restriction onto a cube]			\label{lem: restriction}
  Let $H_1$ and $H_2$ be a pair of parallel affine subspaces of $\R^m$ with equal dimensions.
  Assume that for some scalars $a<b$, we have
  $$
  H_i \cap [a,b]^m \ne \emptyset, \quad i=1,2.
  $$
  Fix any $\l>0$ and consider the cube $Q = [a-\l, b+\l]^m$.
  Then
  $$
  d_\infty \left( H_1 \cap Q,\, H_2 \cap Q \right)
  \le \Big( \frac{b-a}{\l} + 2 \Big) \; d_\infty \left( H_1, H_2 \right).
  $$
\end{lemma}

\begin{proof}
Due to symmetry, it is enough to bound the quantity
$$
\sup_{h_1 \in H_1 \cap Q} \; \inf_{h_2 \in H_2 \cap Q} \; \norm{h_1-h_2}_\infty.
$$
So let us fix any $h_1 \in H_1 \cap Q$ and find $h_2 \in H_2 \cap Q$ for which $\norm{h_1-h_2}_\infty$ is small.
To this end, fix a vector
\begin{equation}	\label{eq: x1}
x_1 \in H_1 \cap [a,b]^m,
\end{equation}
which exists by assumption. Due to the definition of Hausdorff distance, we can find  $x_2 \in H_2$ such that
\begin{equation}	\label{eq: x1-x2}
\norm{x_2-x_1}_\infty \le d_\infty(H_1,H_2) \eqqcolon \d.
\end{equation}
Consider the vector
$$
y \coloneqq x_1 + \frac{\l}{\d}(x_2-x_1)
$$
and set $h_2$ to be the following convex combination of $h_1$ and $y$:
$$
h_2 \coloneqq \Big( 1-\frac{\d}{\l} \Big) h_1 + \frac{\d}{\l} y.
$$
(Here we assume that $\d\leq\l$. Otherwise, the result follows immediately, since the diameter of $Q$ in $L^\infty$-norm is $b-a+2\l$.) Figure~\ref{fig: restriction} might help to visualize our construction.

\begin{figure}[htp]
  \centering
  \includegraphics[width=.45\textwidth]{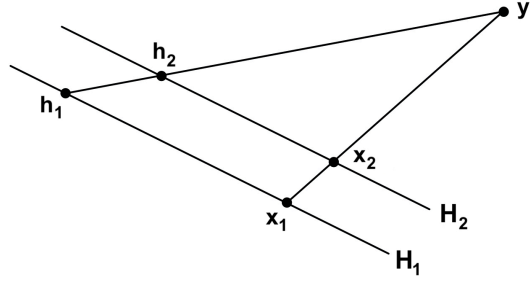}
  \caption{Construction in the proof of Lemma~\ref{lem: restriction}.}
  \label{fig: restriction}
\end{figure}

Let us check that the vector $h_2$ constructed this way satisfies all the required properties. First, we claim that
$$
y \in Q.
$$
Indeed, the definition of $y$ combined with \eqref{eq: x1} and \eqref{eq: x1-x2} yields
$$
y \in [a,b]^m + \frac{\l}{\d} [-\d,\d]^m
= [a-\l,b+\l]^m = Q.
$$
We claim that
$$
h_2 \in H_2.
$$
Indeed, substituting the definition of $y$ into the expression for $h_2$, we get
\begin{equation}	\label{eq: h2 equivalent}
h_2 = \Big( 1-\frac{\d}{\l} \Big) \Big( h_1-x_1 \Big) + x_2
\end{equation}
By the assumption, $H_1$ and $H_2$ are translates of the same linear subspace. This linear subspace can be expressed as $H_1-x_1$ or, equivalently, as $H_2-x_2$ since $x_1 \in H_1$ and $x_2 \in H_2$.
In particular,
we have $t(H_1-x_1) = H_2-x_2$ for any $t \in \R$, or equivalently
$H_2 = t(H_1-x_1) + x_2$.
Since $h_1 \in H_1$, it follows from \eqref{eq: h2 equivalent} that $h_2 \in H_2$ as claimed.

Next, since both $h_1$ and $y$ lie in $Q$, their convex combination must lie there, too, so
$$
h_2 \in Q.
$$

Finally, using the definition of $h_2$ and recalling that $h_1$ and $y$ lie in $Q$, we get
$$
h_1-h_2
= \frac{\d}{\l}(h_1-y)
\in \frac{\d}{\l}(Q-Q)
= \frac{\d}{\l} \left[ -(b-a+2\l), b-a+2\l \right]^m.
$$
Thus
$$
\norm{h_1-h_2}_\infty
\le \frac{\d}{\l} (b-a+2\l)
= \Big( \frac{b-a}{\l} + 2 \Big) \; \d.
$$
The proof if complete.
\end{proof}

\subsection{$L^\infty$ sensitivity of the selection rule}

We are ready to analyze the sensitivity of the $L^2$-proximal selection rule:

\begin{lemma}[$L^{\infty}$ sensitivity of the selection rule]			\label{lem:Linftysensitivity}
  Let $0< a < c < (a+b)/2$.
  Let $H_1$ and $H_2$ be a pair of parallel affine subspaces of $\R^m$ with equal dimensions. Assume that
  $$
  H_i \cap [a,b]^m \ne \emptyset, \quad i=1,2.
  $$
  Let
  $$
  h_i = \argmin \left\{ \norm{x - c \cdot \one_m}_2 :\; x \in H_i \cap [a-\l,b+\l]^m \right\}, \quad i=1,2.
  $$
  Then
  $$
  \norm{h_1-h_2}_\infty^2 \le 4m(b-c) \Big( \frac{b-a}{\l} + 2 \Big) \; d_\infty \left( H_1, H_2 \right).
  $$
  \end{lemma}

\begin{proof}
Lemma~\ref{lem: restriction} gives
\begin{equation}	\label{eq: dinfty K1 K2}
d_\infty(K_1,K_2)
\le \Big( \frac{b-a}{\l} + 2 \Big) \; d_\infty \left( H_1, H_2 \right)
\end{equation}
where $K_i = H_i \cap [a-\lambda,b+\lambda]^m$.
Let us apply Lemma~\ref{lem:proximalinL2} for $r = c \cdot \one_m$ and the $L^2$ norm on $\R^m$. Note that
$$
\dist_{L^2}(r,K_i)
\le \max_{h \in [a,b]^m} \norm{r-h}_{L^2}
\le \max_{h \in [a,b]^m} \norm{r-h}_\infty
= \max \left\{ \abs{a-c}, \abs{c-b} \right\} = b-c.
$$
Thus, Lemma~\ref{lem:proximalinL2} yields
$$
\norm{h_1-h_2}_{L^2}^2
\le 4 (b-c) \cdot d_{L^2}(K_1,K_2)
\le 4 (b-c) \cdot d_\infty(K_1,K_2).
$$
To complete the proof, use \eqref{eq: dinfty K1 K2} and note that
$\norm{h_1-h_2}_\infty^2 \le m \norm{h_1-h_2}_{L^2}^2$.
\end{proof}

\section{Shrinkage}

Another step of Algorithm~\ref{mainalgorithm1} we need to control is shrinkage.
We will check here that shrinkage onto a cube is Lipschitz in the $L^\infty$-Hausdorff metric. Let us start with a general observation:

\begin{lemma}[Shrinkage]				\label{lem: shrinkage}
  Let $X$ be a normed space and $z \in X$ be a point such that
  $\norm{z} \le 1-\b$ for some $\beta \in (0,1)$.
  Let $r:X \to X$ be the retraction map onto the unit ball of $X$ toward $z$, i.e.
  $$
  r(x) = (1-\l) x + \l z
  $$
  where $\l = \l(x)$ is the minimal number in $[0,1]$ such that $\norm{r(x)} \le 1$.
  Then the Lipschitz norm of the map $\l(\cdot)$ is at most $1/\b$,
  and the Lipschitz norm of the map $r(\cdot)$ is at most $2/\b$.
\end{lemma}

\begin{proof}
Fix any pair of vectors $x_1, x_2 \in X$ and denote
$$
\l_1 = \l(x_1), \quad
\l_2=\l(x_2), \quad
\mu=\norm{x_1-x_2}/\b.
$$
The claim about the Lipschitz norm of $\l(\cdot)$ can be stated as
$\abs{\l_1-\l_2} \le \mu$.
By symmetry, it suffices to show that
\begin{equation}	\label{eq: l1 l2 mu}
\l_1 \le \l_2+\mu.
\end{equation}
This bound is trivial if $\l_2+\mu > 1$ since we always have $\l_1 \le 1$.
So we can assume from now on that $\l_2+\mu \in [0,1]$.

Due to the minimality property in the definition of $\l_1 = \l(x_1)$, in order
to prove \eqref{eq: l1 l2 mu} it suffices to show that
\begin{equation}	\label{eq: stricter convex combination}
\norm{(1-\l_2-\mu)x_1+(\l_2+\mu)z} \le 1.
\end{equation}
By triangle inequality, the left hand side is bounded by $\norm{A} + \norm{B}$
where
$$
A = (1-\l_2-\mu)x_2+(\l_2+\mu)z, \quad
B = (1-\l_2-\mu)(x_1-x_2).
$$

Rearranging the terms, we can rewrite
$$
A = (1-a) \left[ (1-\l_2)x_2+\l_2 z \right] + az
\quad \text{where} \quad
a = \frac{\mu}{1-\l_2}.
$$
By assumption, $a \in [0,1]$. Then $A$ is a convex combination of the vector
$(1-\l_2)x_2+\l_2 z$ whose norm is bounded by $1$ by definition of $\l_2=\l(x_2)$
and the vector $z$ whose norm is bounded by $1-\b$ by assumption.
Hence, by triangle inequality and definition of $a$ and $\mu$, we have
$$
\norm{A} \le (1-a) \cdot 1 + a \cdot (1-\b)
= 1-a\b
\le 1 - \mu \b
= 1 - \norm{x_1-x_2}.
$$
Furthermore, the assumption $1-\l_2-\mu \in [0,1]$ yields
$$
\norm{B} \le \norm{x_1-x_2}.
$$
Hence we showed that $\norm{A} + \norm{B} \le 1$,
establishing \eqref{eq: stricter convex combination} and completing
the first part of the proof (about the Lipschitz norm of $\l$).

To prove the second part of the lemma, we need to show that
\begin{equation}	\label{eq: claim}
\norm{r(x_1)-r(x_2)} \le (2/\b) \norm{x_1-x_2}.
\end{equation}

Let us first prove this inequality assuming that $\norm{x_1} \le 1$ or $\norm{x_2} \le 1$.
Without loss of generality, assume $\norm{x_1} \le 1$.
Denoting $\mu_1=1-\l_1$ and $\mu_2=1-\l_2$ and using triangle inequality, we obtain
\begin{equation}	\label{eq: r1 r2}
\norm{r(x_1)-r(x_2)}
= \norm{\mu_1 x_1 + \l_1 z - \mu_2 x_2 - \l_2 z}
\le \norm{\mu_1 x_1 - \mu_2 x_2} + \abs{\l_1-\l_2} \norm{z}
\end{equation}
By the first part of the lemma and since $\norm{z} \le 1-\b$, we have
\begin{equation}	\label{eq: l1 l2 z}
\abs{\l_1-\l_2} \norm{z} \le \frac{1}{\b} \norm{x_1-x_2} (1-\b) = (1/\b-1) \norm{x_1-x_2}.
\end{equation}
Furthermore, adding and subtracting the cross term $\mu_2 x_1$ and using triangle inequality, we get
$$
\norm{\mu_1 x_1 - \mu_2 x_2}
\le \abs{\mu_1-\mu_2} \, \norm{x_1} + \mu_2 \, \norm{x_1-x_2}.
$$
Now, $\abs{\mu_1-\mu_2} = \abs{\l_1-\l_2} \le \norm{x_1-x_2}/\b$ by the first part of the lemma; $\norm{x_1} \le 1$ by the standing assumption, and $\mu_2 \le 1$. Hence
\begin{equation}	\label{eq: m1x1 m2x2}
\norm{\mu_1 x_1 - \mu_2 x_2}
\le (1/\b+1) \norm{x_1-x_2}.
\end{equation}
Substitute \eqref{eq: l1 l2 z} and \eqref{eq: m1x1 m2x2} into \eqref{eq: r1 r2}, we
conclude the claim \eqref{eq: claim}.

Finally, consider the remaining case where both $\norm{x_1} \ge 1$ and $\norm{x_2} \ge 1$. Without loss of generality, $\l_1 \le \l_2$, so the vectors
$$
\tilde{x}_1 \coloneqq (1-\l_1)x_1 + \l_1 z
\quad \text{and} \quad
\tilde{x}_2 \coloneqq (1-\l_1)x_2 + \l_1 z
$$
satisfy
$$
\norm{\tilde{x}_1} = 1
\quad \text{and} \quad
\norm{\tilde{x}_2} \ge 1.
$$
Definition of retraction yields $r(\tilde{x}_1) = r(x_1)$ and $r(\tilde{x}_2) = r(x_2)$.
Thus, applying \eqref{eq: claim} for $\tilde{x}_1$ and $\tilde{x}_2$, we get
\begin{equation*}
\begin{split}
\norm{r(x_1)-r(x_2)}
= \norm{r(\tilde{x}_1)-r(\tilde{x}_2)}
\le (2/\b) \norm{\tilde{x}_1-\tilde{x}_2} \\
= (2/\b) (1-\l_1) \norm{x_1-x_2}
\le (2/\b) \norm{x_1-x_2}.
\end{split}
\end{equation*}
The lemma is proved.
\end{proof}

Now we extend our analysis of shrinkage for affine subspaces:

\begin{lemma}[Shrinkage for subspaces]		\label{lem: shrinkage for subspaces}
  Let $K$ be the unit ball of a finite dimensional normed space $X$.
  Let $z, z_0 \in X$ be points such that
  $z \in z_0 + (1-\b)K$ for some $\beta \in (0,1)$.
  Given an affine subspace $H$ in $X$, define the affine subspace $\tilde{H}$
  by moving $H$ toward $z$ until it intersects the ball $z_0+K$, i.e.
  $$
  \tilde{H} = (1-\l) H + \l z
  $$
  where $\l = \l(H)$ is the minimal number in $[0,1]$ such that
  $\tilde{H} \cap (z_0+K) \ne \emptyset$.
  Then for any two affine subspaces $H_1$ and $H_2$ that are translates of each other,
  the Hausdorff distance satisfies
  $$
  d_X(\tilde{H}_1, \tilde{H}_2) \le \frac{2}{\beta} \, d_X(H_1, H_2).
  $$
\end{lemma}

\begin{proof}
By translation, we can assume without loss of generality that $z_0=0$.
The affine subspaces $H_1$ and $H_2$ are translates of some common linear subspace $H_0$. Apply Lemma~\ref{lem: shrinkage} for the quotient space $X/H_0$ instead of $X$ and for $H_z \coloneqq z+H_0$ instead of $z$.

The requirement of that lemma is satisfied since
\begin{equation}	\label{eq: quotient small}
\norm{H_z}_{X/H_0}
= \inf_{h \in H_z} \norm{h}_X
\le \norm{z}_X
\le 1-\b.
\end{equation}
Indeed, the equality here is the definition of the norm in the quotient space, the first inequality holds since $z \in H_z$, and the last inequality is an equivalent form of the assumption $z \in (1-\b)K$.

We claim that the retraction map $r(\cdot)$ in Lemma~\ref{lem: shrinkage} satisfies
$$
r(H) = \tilde{H}
\quad \text{for any translate $H$ of $H_0$}.
$$
Indeed, by definition we have
$$
r(H) = (1-\l)H + \l H_z
$$
where $\l$ is the minimal number in $[0,1]$ such that
$\norm{r(H)}_{X/H_0} \le 1$. Since $\norm{H_z}_{X/H_0} < 1$ by \eqref{eq: quotient small}, continuity shows that $\l<1$ and hence
$$
r(H) = (1-\l)H + \l z.
$$
Moreover, the condition that $\norm{r(H)}_{X/H_0} \le 1$ is equivalent to
$r(H) \cap K \ne \emptyset$. Hence the definitions of $r(H)$ and $\tilde{H}$ are equivalent as we claimed.

Lemma~\ref{lem: shrinkage} yields
$$
\norm[1]{\tilde{H}_1-\tilde{H}_2}_{X/H_0} \le \frac{2}{\beta} \, \norm{H_1-H_2}_{X/H_0}.
$$
It remains to note that, by definition,
$$
\norm{H_1-H_2}_{X/H_0}
= \inf_{h_1 \in H_1, \, h_2 \in H_2} \norm{h_1-h_2}_X
= d_X(H_1, H_2),
$$
and similarly for the distance between $\tilde{H}_1$ and $\tilde{H}_2$.
The proof is complete.
\end{proof}

Finally, we specialize our analysis to the shrinkage onto the cube:

\begin{lemma}[Shrinkage onto a cube]			\label{lem: shrinkage onto cube}
  Let $0< a < c < (a+b)/2$.
  Given an affine subspace $H$ in $\R^m$, define the affine subspace $\tilde{H}$
  by moving $H$ toward $d \one_m$ until it intersects the cube $[a,b]^m$, i.e.
  $$
  \tilde{H} = (1-\l) H + \l \cdot c \one_m
  $$
  where $\l = \l(H)$ is the minimal number in $[0,1]$ such that
  $\tilde{H} \cap [a,b]^m \ne \emptyset$.
  Then for any two affine subspaces $H_1$ and $H_2$ that are translates of each other,
  the Hausdorff distance in the $L^\infty$ norm satisfies
  $$
  d_\infty(\tilde{H}_1, \tilde{H}_2) \le \frac{b-a}{c-a} \, d_\infty(H_1, H_2).
  $$
\end{lemma}

\begin{proof}
Apply Lemma~\ref{lem: shrinkage for subspaces} for
$$
z = c \one_m, \quad
z_0 = \frac{a+b}{2} \, \one_m, \quad
K = \Big[ -\frac{b-a}{2}, \, \frac{b-a}{2} \Big]^m.
$$
so that $z_0$ is the center of the cube $[a,b]^m$, $K$ is the centered cube, and
$z_0+K = [a,b]^m$.

Now,
$$
z-z_0 = \Big( c - \frac{a+b}{2} \Big) \, \one_m
$$
and
$$
0 \le \frac{a+b}{2} - c = (1-\b) \frac{b-a}{2}
\quad \text{for } \beta = \frac{2(c-a)}{b-a},
$$
so
$z - z_0 \in (1-\b) K$
as required in Lemma~\ref{lem: shrinkage for subspaces}.
The conclusion of this lemma is that
$$
d_X(\tilde{H}_1, \tilde{H}_2) \le \frac{2}{\beta} \, d_X(H_1, H_2).
$$
Since the unit ball $K$ of $X$ is the cube $[-1,1]^m$
scaled by the factor $(b-a)/2$, the norm in $X$
is the $L^\infty$-norm scaled by that factor. Therefore, the conclusion
holds for the $L^\infty$ norm as well.
\end{proof}

\section{Privacy and accuracy of the algorithm}

We are ready to analyze the privacy and accuracy of Algorithm~\ref{mainalgorithm1}.

\subsection{Algorithm}

For convenience we  rewrite Algorithm~\ref{mainalgorithm1}, see Algorithm~\ref{mainalgorithm2} below.
Also note that in Step 5 of Algorithm~\ref{mainalgorithm2}, the $L^2(S)$-norm is defined as $\norm{h}_{L^2(S)}^2 = \frac{1}{m} \sum_{i=1}^m h(s_i)^2$.

\begin{algorithm}[h!]
\caption{Private sampling synthetic data algorithm}
\label{mainalgorithm2}
\begin{algorithmic}

\State {\bf Input: } a sequence $X$ of $n$ points in  $\{-1,1\}^p$ (true data); $m$:  cardinality of $S$;
$d$: the degree of the marginals to be matched; parameters $\delta, \Delta$ with $\Delta>  \delta > 0$.

\begin{enumerate}[1.]

\item Draw a sequence $S = (\theta_1,\ldots,\theta_m)$ of $m$ points in the cube independently and uniformly (reduced space).

\item Form the $m \times \binom{p}{\le d}$ matrix $M$ with entries $w_J(\theta_i)$,
i.e. the matrix whose rows are indexed by the points of the reduced space $S$
and whose columns are indexed by the Walsh functions of degree at most $d$.
If the smallest singular value of $M$ is bounded below by $\sqrt{m}/2e^d$, call $S$ well conditioned and proceed. Otherwise return ``Failure'' and stop.

\item Let $f_n$ be the uniform density on true data: $f_n = \frac{1}{n} \sum_{i=1}^n \one_{x_i}$.
Consider the solution space
$$
H = H(f_n) = \left\{ h:\{-1,1\}^p \to \R:\; \supp(h) \subset S, \; h^{\le d} = (f_n)^{\le d} \right\},
$$

\item Shrink $H$ toward the uniform density $u_m = \frac{1}{m} \sum_{i=1}^m \one_{s_i}$ on $S$: let
$$
\tilde{H} = (1-\l)H + \l u_m
$$
where $\l \in [0,1]$ is the minimal number such that $\tilde{H} \cap [2\d/m, (\Delta-\d)/m]^S \ne \emptyset$.

\item Pick a proximal point

$$
h^* = \argmin \left\{ \norm[1]{\tilde{h}-u_m}_{L^2(S)}
	:\; \tilde{h} \in \tilde{H} \cap [\d/m, \Delta/m]^S \right\}.
$$

\end{enumerate}

\State {\bf Output:} a sequence $Y = (y_1,\ldots,y_k)$ of $k$ independent points
drawn from $S$ according to density $h^*$.

\end{algorithmic}
\end{algorithm}

\medskip

The standing assumption in this section is that the reduced space $S = (s_1,\ldots,s_m)$ is random, and consists of points $s_i$ drawn independently and uniformly from the cube. We would like to show that with high probability over $S$, the algorithm is differentially private.

\subsection{Sensitivity of density}

The privacy guarantee will be achieved via Private Sampling Lemma~\ref{lem: private sampling}. To apply it, we need to bound the sensitivity of the density $h^*$ computed by the algorithm.

\begin{lemma}			\label{lem: sensitivity of density}
  Suppose the reduced space $S$ is well conditioned.
  Then, for any pair of input datasets $X_1$ and $X_2$ that consist
  of at least $n$ elements each and differ from each other by a single element,
  the densities $h^*_1$ and $h^*_2$ computed by the algorithm satisfy
  $$
  \norm{h^*_1 - h^*_2}_\infty
  \le \frac{4\sqrt{2} \Delta^{3/2} e^{d/2}}{\sqrt{\d n} \, m^{1/4}} \binom{p}{\le d}^{1/4}.
  $$
\end{lemma}

\begin{proof}
By Proposition~\ref{prop: translates of each other}, the solution subspaces
$$
H_1 = H(f_n)
\quad \text{and} \quad
H_2 = H(f_{n+1})
$$
are translates of each other. The ambient space consists of all functions supported on an $m$-element set $S$, and thus can be identified with $\R^m$.
Let $\tilde{H}_i$ be the result of shrinkage of the subspaces $H_i$ toward the uniform distribution as specified in the algorithm, i.e. the shrinkage onto the cube $[\d/m, \Delta/m]^m$ and toward the uniform distribution $u_m$.
The selection rule for $h^*$ specified in the algorithm is stable in the $L^\infty$ metric.  Indeed, Lemma~\ref{lem:Linftysensitivity} applied for the subspaces $\tilde{H}_i$ and for
$$
a = \frac{2\d}{m}, \quad
b = \frac{\Delta-\d}{m}, \quad
c = \frac{1}{m}, \quad
\l = \frac{\d}{m}
$$
yields
$$
\norm{h^*_1 - h^*_2}_\infty^2
\le \frac{4\Delta^2}{\d} \cdot d_\infty(\tilde{H}_1,\tilde{H}_2).
$$
Next, recall that the shrinkage map is stable.
Indeed, Lemma~\ref{lem: shrinkage onto cube} applied for the same $a,b,c$ yields
$$
d_\infty(\tilde{H}_1,\tilde{H}_2) \le 2\Delta \cdot d_\infty(H_1,H_2).
$$
Furthermore, the solution space is stable. Indeed, Lemma~\ref{lem: sensitivity of solution space} for the uniform density $\mu$ on the cube yields
$$
d_\infty(H_1,H_2)
\le \frac{2e^d2^p}{\sqrt{m}} \, \norm[1]{(f_n-f_{n+1})^{\le d}}_{L^2}.
$$
Finally, recall from \eqref{eq: fn+1 fn} that
$$
\norm[1]{(f_{n+1} - f_n)^{\le d}}_{L^2}
\le \frac{2}{n} \binom{p}{\le d}^{1/2} \frac{1}{2^p}.
$$
Combining all these bounds, we conclude that
$$
\norm{h^*_1 - h^*_2}_\infty^2
\le \frac{4\Delta^2}{\d} \cdot 2\Delta \cdot \frac{2e^d2^p}{\sqrt{m}} \cdot \frac{2}{n} \binom{p}{\le d}^{1/2} \frac{1}{2^p}
\le \frac{32\Delta^3 e^d}{\d n \sqrt{m}} \binom{p}{\le d}^{1/2}.
$$
The proof is complete.
\end{proof}

\subsection{Privacy guarantee}

Finally, we are ready to give the privacy guarantee of our algorithm:

\begin{theorem}[Privacy]				\label{thm: privacy}
  If $k \le \frac{1}{4\sqrt{2}} \e (\d/\Delta)^{3/2} e^{-d/2} \binom{p}{\le d}^{-1/4} \sqrt{n}/m^{3/4}$,
  then Algorithm~\ref{mainalgorithm2} is $\e$-differentially private.
\end{theorem}

\begin{proof}
Since the reduced space $S$ is drawn independently of the input data $X$,
we can condition on $S$. If $S$ is ill conditioned, the algorithm returns ``Failure'' regardless of the input data, so the privacy holds trivially. Suppose $S$ is well conditioned.

Let $X_1$ and $X_2$ be a pair of datasets that consist
of at least $n$ elements each and differ from each other by a single element.
By the choice made in the algorithm and by sensitivity of density (Lemma~\ref{lem: sensitivity of density}), we have
$$
h^*_2 \ge \frac{\d}{m}
\quad \text{and} \quad
\abs{h^*_1 - h^*_2}
\le \frac{4\sqrt{2} \Delta^{3/2} e^{d/2}}{\sqrt{\d n} \, m^{1/4}} \binom{p}{\le d}^{1/4}
\eqqcolon \eta
$$
pointwise. Therefore
$$
\abs{h^*_1/h^*_2}
\le 1 + \frac{\eta m}{\d}
\le \exp \Big( \frac{\eta m}{\d} \Big)
\le \exp \Big( \frac{\e}{k} \Big)
$$
pointwise, where the last inequality indeed holds due to our assumption on $k$.
Private Sampling Lemma~\ref{lem: private sampling} completes the proof.
\end{proof}

\begin{remark}		\label{rem: DP simplified}
  Suppose we chose the size $m$ of the reduced space $S$ so that
  $m \asymp e^{2d} \binom{p}{\le d}$. Simplifying the condition in Theorem~\ref{thm: privacy},  we conclude that if $k \ll \sqrt{n} /m$, then Algorithm~\ref{mainalgorithm2} is $o(1)$-differentially private.
  \end{remark}


\subsection{Accuracy guarantee}

The following is the accuracy guarantee of our algorithm:

\begin{theorem}[Accuracy]				\label{thm: accuracy}
  Assume the true data $X=(x_1,\ldots,x_n)$ is drawn independently
  from the cube according to some density $f$, which satisfies
  $\norm{f}_\infty \le \Delta/2^p$.
  Assume that $n \ge 16 \d^{-2}\g^{-1} e^{2d} \binom{p}{\le d}$,
  $16 \d^{-2}\g^{-1} \Delta^{2} e^{2d} \binom{p}{\le d} \le m \le 2^{p/4}$,
  and
  $k \ge 4 \d^{-2} (\log(2/\gamma) + \log \binom{p}{\le d})$.
  Then, with probability at least $1-4\g - \frac{1}{\sqrt{2^{p}}}$,
  the algorithm succeeds, and
  all marginals of the synthetic data $Y$ up to dimension $d$ are
  within $4\d$ from the corresponding marginals of the true data $X$.
\end{theorem}

\begin{proof}
Proposition~\ref{prop: success} and the choice of $m$ guarantee
that the algorithm succeeds with probability at least $1-\g$.

Furthermore, the uniform density on the cube $g=2^{-p}$ satisfies
$\norm{f/g}_{L^2} \le \norm{f/g}_\infty = \norm{f}_\infty \cdot 2^p \le \Delta$.
Therefore, Theorem~\ref{thm: master} implies that with probability at least $1-2\g$, there exists $h \in H = H(f_n)$ such that
\begin{equation}	\label{eq: h approximation}
\norm{h - (f/g) g_m}_\infty \le \frac{\d}{m}.
\end{equation}
Since $(f/g) g_m$ is a nonnegative function, it follows that
$$
h \ge -\frac{\d}{m}
\quad \text{pointwise}.
$$
The assumption $m \le 2^{p/4}$ implies that with probability $1- \frac{1}{\sqrt{2^{p}}}$ there are no repetitions in $y_{1},\ldots,y_{m}$, which in turn implies
that with probability $1- \frac{1}{\sqrt{2^{p}}}$ we have $\norm{g_m}_\infty \le 1/m$  (otherwise $\norm{g_m}_\infty$ would scale with the number of repetitions in 
$y_{1},\ldots,y_{m}$).

In the following we condition on the event that there are no repetitions in $y_{1},\ldots,y_{m}$.
Since $\norm{f/g}_\infty \le \Delta$ by above and $\norm{g_m}_\infty \le 1/m$, we have 
$\norm{(f/g) g_m}_\infty \le \Delta/m$, so
$$
h \le \frac{\Delta+\d}{m}
\quad \text{pointwise}.
$$
A combination of these two bounds on $h$ implies that
$$
\frac{2\d}{m} \le (1-3\d)h + \frac{3\d}{m}
\le \frac{\Delta-\d}{m}
\quad \text{pointwise},
$$
as long as $\Delta \ge 5/3$.
Since $h \in H$, it follows that the affine subspace $(1-3\d)H + 3\d u_m$ has a nonempty intersection with $[2\d/m, (\Delta-\d)/m]^m$. The minimality property of $\l$ in the algorithm yields
\begin{equation}	\label{eq: lambda small}
\l \le 3\d.
\end{equation}

Recall that a marginal of a function $f: \{-1,1\}^p \to \R$
that corresponds to a subset $J \subset [p]$ of parameters and
values $\theta_j \in \{-1,1\}$ for $j \in J$, is defined as
$$
P(f) = \sum_{x \in \{-1,1\}^p} f(x) v(x)
$$
where $v(x) = \one_{\{x(j) = \theta_j \; \forall j \in J\}}$.

Recall that the solution $h^*$ of the algorithm satisfies
$$
h^* \in \tilde{H} = (1-\l)H + \l u_m
$$
and, by definition of $H$, all members of $H$ have the same marginals up to dimension $d$ as $f_n$.
This and linearity implies that for any marginal up to dimension $d$,
$$
P(h^*) = (1-\l) P(f_n) + \l P(u_m)
$$
Hence
$$
\abs{P(h^*) - P(f_n)} \le \l \abs{P(u_m) - P(f_n)}
$$
Since $u_m$ and $f_n$ are densities, all of their marginals must be within $[0,1]$, so
$\abs{P(u_m) - P(f_n)} \le 1$. Combining this with \eqref{eq: lambda small}, we get
\begin{equation}	\label{eq: Phstar Pfn}
\abs{P(h^*) - P(f_n)} \le 3\d,
\end{equation}
for all marginals up to dimension $d$, with probability at least $1-2\gamma$.

Now we compare the marginals of the density $h^*$ and its empirical counterpart $h^*_k$. We can express
$$
P(h^*_k) - P(h^*) = \frac{1}{k} \sum_{i=1}^k \left( v(Y_i) - \E v(Y_i) \right)
$$
where $Y_i$ are i.i.d. random variables drawn according to the density $h^*$.
Thus, we have a normalized and centered sum of i.i.d. Bernoulli random variables, so Bernstein's inequality (see e.g. \cite[Theorem~2.8.4]{vershyninbook}) yields
$$
\Pr{\abs{P(h^*_k) - P(h^*)} > \delta} \le 2 \exp(- \d^2 k/4) \le \g \binom{p}{\le d}^{-1}
$$
if $k \ge 4 \d^{-2} (\log(2/\gamma) + \log \binom{p}{\le d})$.
Thus, by a union bound, we have
$$
\abs{P(h^*_k) - P(h^*)} \le \d,
$$
simultaneously for all marginals up to dimension $d$, with probability at least $1-\gamma$.

Combining this with \eqref{eq: Phstar Pfn} via the triangle inequality, we conclude that
$$
\abs{P(h^*_k) - P(f_n)} \le 4\d,
$$
for all marginals up to dimension $d$, with probability at least $1-3\gamma$. Recalling that we conditioned on an event with probability
$1- 1/\sqrt{p}$ and applying the union bound completes the proof.
\end{proof}

\begin{remark}[No shrinkage for regular densities]		\label{rem: no shrinkage}
  If the density $f$ from which the true data $X$ is drawn is regular, specifically if $3\d/2^{p} \le f \le (\Delta-2\d)/2^p$ pointwise for some positive numbers $\delta$ and $\Delta$, {\em the algorithm does not apply any shrinkage}. Indeed, in this case
  we have $3\d/m \le (f/g)g_m \le (\Delta-2\d)m$, so it follows from \eqref{eq: h approximation} that $2\d/m \le h \le (\Delta-\d)m$, and thus $H$ has a nonempty intersection with $[2\d/m, (\Delta-\d)m]^S$, hence $\l=0$.
\end{remark}

\section*{Acknowledgement}

M.B. acknowledges support from NSF DMS-2140592. T.S. acknowledges support from NSF-DMS-1737943, NSF DMS-2027248, NSF CCF-1934568 and a CeDAR Seed grant.
 R.V. acknowledges support from NSF DMS-1954233, NSF DMS-2027299, U.S. Army 76649-CS, and NSF+Simons Research Collaborations on the Mathematical and Scientific Foundations of Deep Learning.

\end{document}